\documentclass[10pt, conference, compsocconf]{IEEEtran}

\usepackage{amsthm}
\usepackage{amsfonts}
\usepackage{amssymb}
\usepackage{amsmath}
\usepackage{algorithm}
\newtheorem{mydef}{Definition}

\newtheorem{thm}{Theorem}

\newtheorem{cor}{Corollary}
\newtheorem{prop}{Property}
\usepackage{url}
\usepackage{graphicx}
\usepackage{setspace}

\ifCLASSINFOpdf
\else
\fi
\usepackage{algorithmic}
\usepackage{multirow}

\usepackage{cite}

\begin{document}
\title{On the Complexity of Sorted Neighborhood}
\author{\IEEEauthorblockN{Mayank Kejriwal}
\IEEEauthorblockA{University of Texas at Austin\\
kejriwal@cs.utexas.edu}
\and
\IEEEauthorblockN{Daniel P. Miranker}
\IEEEauthorblockA{University of Texas at Austin\\
miranker@cs.utexas.edu}
}

\maketitle
\fontsize{10pt}{10.2pt}
\selectfont

\begin{abstract}
Record linkage concerns identifying semantically equivalent records in databases. Blocking methods are employed to avoid the cost of full pairwise similarity comparisons on $n$ records. In a seminal work, Hern{\`a}ndez and Stolfo proposed the Sorted Neighborhood blocking method. Several empirical variants have been proposed in recent years. In this paper, we investigate the complexity of the Sorted Neighborhood procedure on which the variants are built. We show that achieving maximum performance on the Sorted Neighborhood procedure entails solving a sub-problem, which is shown to be NP-complete by reducing from the Travelling Salesman Problem. We also show that the sub-problem can occur in the traditional blocking method. Finally, we draw on recent developments concerning approximate Travelling Salesman solutions to define and analyze three approximation algorithms.            
\end{abstract}


\IEEEpeerreviewmaketitle

\section{Introduction}\label{introduction}
Record linkage concerns identifying pairs of records that refer to the same underlying entity but are \emph{syntactically} disparate. The problem goes by multiple names in the database community, examples being \emph{entity resolution} \cite{swoosh}, \emph{instance matching} \cite{ERsurvey},  \emph{co-reference resolution} \cite{coref}, \emph{hardening soft} databases \cite{hardening} and the \emph{merge-purge} problem \cite{stolfo}. 

Given $n$ records and a sophisticated \emph{similarity} function $g$ that determines whether two records are  equivalent, a na{\"i}ve record linkage application would run in time $\Theta(t(g)n^2)$, where $t(g)$ is the run-time of $g$. Scalability indicates a two-step approach \cite{recordlinkagesurvey}. First, \emph{blocking} generates a \emph{candidate set} of promising pairs that have the potential to be duplicates. The vast majority of pairs are discarded in this step, leading to significant savings \cite{christensurvey}. Because of the need to limit complexity of record linkage to \emph{near} linear-time, blocking has emerged as a research area in its own right \cite{datamatching, christensurvey}.   

\emph{Sorted Neighborhood} is a popular blocking method published originally by Hern{\`a}ndez and Stolfo \cite{stolfo}. The method was found to have excellent empirical performance \cite{stolfo,stolfo98}. In the past two decades, numerous empirical variants have been published \cite{windowing, adaptive}, including an application to XML duplicate detection \cite{xml}. Parallel implementations also continue to be researched. For example, a  \emph{MapReduce}-based implementation of Sorted Neighborhood was published in 2012 \cite{mapreduce, snmapreduce}. The evidence indicates that Sorted Neighborhood remains topical in the data matching community. 

Table \ref{example1}, used as a running example throughout the paper, illustrates the original method. First, a \emph{blocking key} is defined, and applied on each record to generate a \emph{blocking key value} or BKV for the record. In Table \ref{example1}, it is defined as extracting and concatenating initial characters from attribute tokens in the record in order to generate the BKV. Each record's BKV has also been noted in Table \ref{example1}. Next, the BKVs are used as \emph{sorting keys}. Finally, a window of constant size $w \geq 2$ is slid over the sorted records from beginning to end, with records sharing a window paired, and the pair added to the candidate set. With $w=2$, for example, record pair\footnote{$r_i$ refers to record with ID $i$, $i \in \{1,2,\ldots , 7\}$} ($r_1$,$r_2$) would get added to the (initially empty) candidate set. Sliding the window forward, record pair ($r_2$,$r_3$) is added. The process continues, with record pair ($r_6$,$r_7$) being the final addition to the set. 

We define the \emph{w-ordering problem} as the problem of sorting records that have the \emph{same} BKVs, as in the case of records $r_1$,$r_2$ and $r_3$. The definition is formally given as Definition \ref{wordering} in Section \ref{windowing}. Suppose that a polynomial-time \emph{scoring heuristic} $f$ is provided, such that $f$ returns a real-valued similarity score for a given record pair. The goal is to order records so as to maximize the score of the resulting candidate set for given $f$ and $w$, but without violating the sorting order imposed by the BKVs. 
Assuming $w=2$ and a heuristic based on first and last name similarity, candidate set score will be maximized for the ordering in Table \ref{example1}. Table \ref{example1} is then said to be a \emph{maximum-score} 2-ordering for the corresponding \emph{set} of records $\{r_1,r_2,\ldots , r_7\}$. As an example of a 2-ordering that's not maximum-score, consider reversing the positions of $r_5$ and $r_6$. In this case, the reversal would cause \emph{two} potentially duplicate pairs to get left out of the candidate set, but which contributed scores to the candidate set earlier. 
 
Herein it is shown that, in the general case, achieving a maximum-score w-ordering for a set of records is NP-complete. A Karp reduction from the NP-complete Travelling Salesman Problem (TSP) is presented (Theorem \ref{2optim}, Section \ref{windowing}) \cite{TSPpath}. 
To the best of our knowledge, w-ordering has not been studied in previous literature on blocking methods. A possible explanation is that practitioners assumed a \emph{random} ordering with a large window size $w$ to yield a good empirical solution, especially for small datasets. As Hern{\`a}ndez and Stolfo found in their own experiments, such an approach is outperformed by a multi-pass Sorted Neighborhood approach with inexpensive blocking keys, and small window sizes \cite{stolfo}. In particular, 3 passes, the minimum window size of 2  and transitive closure in the second record linkage step was found to achieve a good balance of run-time and accuracy on a test database\footnote{As evidence, we refer the reader to Figure 4 and the conclusion in the original paper \cite{stolfo}}.

Given these findings and that the run-time of a multi-pass approach is proportional to both $w$ and the number of runs \cite{stolfo}, we argue that refining Sorted Neighborhood further even for $w=2$ is an important problem. A review of TSP literature shows that improved approximation bounds for the \emph{max tour-TSP} variant continue to be proposed \cite{maxTSP}. By reducing maximum-score 2-ordering \emph{to} max TSP, we devise three polynomial-time approximation algorithms for maximum-score 2-ordering. The goal is to improve theoretical SN performance by presenting tractable, bounded approximations for maximum-score 2-ordering.      

Two of the three proposed algorithms present multi-pass Sorted Neighborhood but with approximate solutions to maximum-score 2-ordering integrated into the procedure. A third algorithm presents a similar solution for \emph{traditional blocking} \cite{christensurvey} in the \emph{MapReduce} paradigm \cite{mapreduce}. We present this case for two reasons. First, the case shows that solutions to the ordering problem need not be restricted to Sorted Neighborhood, but potentially apply to other popular blocking methods as well. Secondly, it demonstrates that 2-ordering solutions do not necessitate serial architectures.  

The three algorithms invoke a max TSP subroutine as a black box, and their qualitative performance is shown to closely mirror that of the invoked subroutine. This implies that further improvements in max TSP bounds directly lead to similar improvements in the algorithms. Using current TSP results, the bound is 61/81 for arbitrary non-negative edge weight functions \cite{maxTSP}. We devise an appropriate reduction and show that two of our three algorithms have \emph{exactly} this bound.

We summarize run-time and quality results for all three algorithms for both the uniform and Zipf distribution \cite{zipf} of BKVs in a principled fashion. Both distributions are known to occur commonly in practice and were recently used  in a related analytical work on blocking \cite{christensurvey}.

\begin{table}
\caption{Records sorted using blocking key values (BKVs)}
    \begin{tabular}[t]{ |p{0.4cm}|p{1.8cm} | p{1.8cm} | p{1.1cm} || p{1.2cm} |}
    \hline
    {\bf ID}& {\bf First Name} & {\bf Last Name} & {\bf Zip} & {\bf BKV} \\ \hline
1&Cathy  & Ransom & 77111 & CR7   \\ \hline
2&Catherine  & Ridley & 77093 & CR7   \\ \hline
3&Cathy  & Ridley & 77093 & CR7   \\ \hline
   4&John & Rogers  & 78751 & JR7  \\ \hline
5&J. & Rogers & 78732 & JR7  \\ \hline
6&John  & Ridley & 77093 & JR7   \\ \hline
7&John & Ridley Sr. & 77093 & JRS7 \\ \hline
    \end{tabular}
\label{example1}
\end{table}

The outline of the paper is as follows. Section \ref{relatedwork} describes related work, and Section \ref{prelims} describes preliminaries. Section \ref{windowing} defines the w-ordering problem and Section \ref{approximation} presents approximation algorithms. Section \ref{futurework} lists two conjectures and concludes the work. 

\section{Related Work}\label{relatedwork}
\subsection{Record Linkage}
As a problem first noted over five decades ago by Newcombe et al. \cite{newcombe}, record linkage has been the focus of efforts in structured, semistructured and unstructured data communities \cite{recordlinkagesurvey}. It is common to separate efforts in the unstructured data community, where the problem is commonly called \emph{co-reference} or \emph{anaphora resolution} \cite{coref}, from those in the structured and semistructured data communities \cite{recordlinkagesurvey}. In the late 1960s, Fellegi and Sunter placed record linkage in a Bayesian framework \cite{fellegi}, and the model continues to guide state-of-the-art research, which is also influenced heavily by contemporary research in the AI community \cite{winkler}. For example, rule-based approaches were popular during the 1980s and 1990s \cite{rulebased}, but machine learning methods have gained prominence in the last decade \cite{marlin}. Three recent surveys are by Elmagarmid et al. \cite{recordlinkagesurvey}, K{\"o}pcke and Rahm \cite{rahmsurvey} and Winkler \cite{winkler}. A generic, powerful framework that addresses some of the challenges of modern record linkage, both in theory and practice, is \emph{Swoosh} \cite{swoosh}. Several open-source toolkits implementing record linkage techniques are available to the practitioner; we list SecondString \cite{secondstring} and Febrl \cite{febrl} as good examples.

Some alternate record linkage models have recently become popular, including \emph{collective} record linkage \cite{collective1},\cite{collective2} and \emph{iterative} record linkage \cite{harra}.
The problem is also important in the \emph{linked data} and Semantic Web community \cite{bernerslinked}, owing to documented growth of linked open data (LOD\footnote{\url{linkeddata.org}}). Many techniques originally developed for relational databases are being adapted for LOD, including rule-based and machine learning approaches \cite{silk}, \cite{knofuss}. A full survey on Semantic Web record linkage systems was provided by Ferraram et al. \cite{ERsurvey}. Other applications of record linkage include data integration \cite{lenzerini}, knowledge graph identification \cite{kgi}, and biomedical linkage \cite{febrl}. 

Given the expense of record linkage, blocking was recognized as an important preprocessing step even when the problem first emerged \cite{newcombe}. The \emph{traditional blocking} method, which is similar to hashing, continues to be popular \cite{christensurvey}. The Sorted Neighborhood method was proposed in the 1990s, and as noted in Section \ref{introduction}, continues to be used and adapted due to its impressive empirical performance \cite{stolfo}. Christen compares important blocking methods in his survey \cite{christensurvey}, in which he also verifies the good performance of both Sorted Neighborhood and traditional blocking.
We note that, while several empirical variants of Sorted Neighbhorhood exist, all of them rely on the fundamental procedure that was first described in the original paper \cite{stolfo}. The procedure will be formally characterized in Section \ref{windowing}.

Finally, parallel and distributed techniques for record linkage are an active area of research \cite{pswoosh, dswoosh}. MapReduce has emerged as an important paradigm, owing to its documented advantages; we refer the reader to the original paper for an excellent introduction \cite{mapreduce}.

For a synthesis of the multiple threads of record linkage research, we refer the reader to the recently published data matching text by Christen \cite{datamatching}.

\subsection{Travelling Salesman Problem (TSP)}\label{rw-tsp}

Complexity proofs in this paper mainly rely on the Travelling Salesman Problem (TSP), which is among the oldest and best studied NP-complete problems \cite{algorithms}. The classic version of the problem, proposed at least as early as 1954 \cite{tspfirst}, is \emph{min tour-TSP}. Specifically, assume a weighted, undirected and complete graph $G=(V,E,W)$ with arbitrary edge weights. The problem is to locate a minimum-weight Hamiltonian cycle. Even with weights set to either 0 or 1, min tour-TSP was shown to be NP-complete, by virtue of a Karp reduction from the Hamiltonian cycle problem \cite{algorithms}. TSP for directed graphs (also known as \emph{asymmetric} TSP) was also shown to be NP-complete \cite{asymmetric}. In this paper, we only consider symmetric variants and undirected graphs. 

Many variants of TSP have since been shown to be NP-complete, including for weight functions that are metric \cite{christofides} or even Euclidean \cite{euclidean}. Two variants of importance herein are the \emph{min path-TSP} and \emph{max tour-TSP} variants with arbitrary \emph{non-negative} weights, both of which will be described in Section \ref{tsp}.

We note that not all TSP variants are equal from an \emph{approximability} perspective. Define the weight\footnote{We uniformly use the word \emph{weight} instead of \emph{cost} (or \emph{score}) since both min and max optimization problems are considered in this paper} of a tour-TSP solution to be the sum of weights of all edges in the tour; the weight of a path-TSP solution can be similarly defined.
Let the weight of an optimal min tour-TSP solution be $\phi^*$.
A polynomial-time $\rho$-approximation algorithm is an algorithm that is guaranteed to find a solution with weight at most (or \emph{at least}, for \emph{max} variants) $\rho \phi^*$, where $\rho$ is a constant and is denoted as the \emph{approximation ratio} \cite{rho}. Note that $\rho \geq 1$ for min variants and $\rho \leq 1$ for max variants. It is known that for min tour-TSP with  arbitrary weights, a $\rho$-approximation algorithm does not exist unless $P=NP$. However, a $\rho$-approximation algorithm exists for max tour-TSP with arbitrary \emph{non-negative} weights \cite{maxTSP}, and also for min tour-TSP if the weights are \emph{metric} \cite{christofides}. 

The first approximation scheme proposed for metric min tour-TSP was by Christofides, with $\rho=3/2$ \cite{christofides}. The difficulty of TSP is attested to by the fact that this approximation ratio is yet to be improved. Fortunately, approximation ratios continue to be updated for max tour-TSP, as described in Section \ref{tsp}. For a full discussion of TSP, we refer the reader to the seminal book on the subject by Reinelt \cite{tspbook}. In their text, Cormen et al. provide a thorough introduction to the general topics of NP-completeness and approximations \cite{algorithms}. 
 
\section{Preliminaries}\label{prelims}
\subsection{Problem Setting}

The \emph{relational data model} is assumed in this paper, with a brief formalism reproduced for completeness.
A relational database schema $S'$ is a finite set of relation \emph{names}. Each individual name $R' \in S'$ is associated with a set of \emph{attributes}.  
An \emph{instance} $S$ of schema $S'$ assigns to each $R' \in S'$, a finite set $R \in S$ of records. For each attribute in the attribute set of $R' \in S'$, a record in $R \in S$ either has an \emph{attribute value} or NULL, which is a reserved keyword used to indicate missing or non-existent attribute values.  

In this paper, we assume that $S=\{R\}$ and $S'=\{R'\}$. In other words, a single instance $R$ is assumed, with name $R'$ and $m \geq 1$ attributes. 
A single schema is a standard assumption in much of existing record linkage literature  \cite{recordlinkagesurvey}. The original Sorted Neighborhood paper additionally assumed only a single instance \cite{stolfo}. Typically, if more than one instance is expected, possibly with different schemas, a \emph{schema integration} step must be incorporated into the pipeline \cite{schemasurvey}.

We also assume that the number of attributes (or columns) is much smaller than the number of records (or rows), and that `processing' a record takes constant-time. Three real-world examples of such processing are counting tokens in a record, generating token initials (as in Table \ref{example1}), and generating token n-grams. Both assumptions above are standard in the blocking community when analyzing blocking methods \cite{christensurvey}.

\subsection{Travelling Salesman Problem (TSP) variants}\label{tsp} 
In Section \ref{rw-tsp}, we noted that the symmetric TSP variants take as input a \emph{complete}, \emph{undirected} and \emph{weighted} graph $G=(V,E,W)$ without self-loops. Define a Hamiltonian path as a path that includes every vertex in the graph exactly once \cite{algorithms}. 

Define the problem of finding a minimum-weight Hamiltonian path in $G$ as the min path-TSP \cite{TSPpath}. The \emph{decision} version of the problem instance aditionally accepts an integer $k$, and needs to determine if a Hamiltonian path with cost \emph{at most} $k$ exists. Three versions of the problem have been studied, and all are NP-complete \cite{hoogeveen}. In the first version, which is of primary concern in this paper, path endpoints are not specified and the algorithm returns \emph{True} if there exists \emph{any} Hamiltonian path with cost at most $k$. In the other two versions, one or both endpoints are respectively given. These versions are therefore more constrained than the first version.

Hoogeveen adapted Christofides' cubic 3/2-approximation algorithm for all three versions, and showed that the 3/2 bound held for the first two versions \cite{hoogeveen}. For the third version (both endpoints specified), he showed a 5/3 bound\footnote{This surprising result showed that finding a constrained path is \emph{harder} than finding a tour, from an approximability perspective \cite{hoogeveen}}. In 2012, An et al. improved the 5/3 bound to $\frac{1+\sqrt{5}}{2}$ \cite{2012}. In the most recent work we are aware of, Seb\H{o} improved this bound even further to 8/5 \cite{2013}. Hoogeveen's original conclusion on the difficulty of the problem (compared to the tour problem) still stands, since all these bounds are greater than 3/2, which has yet to be improved, to the best of our knowledge.  

The \emph{max} version of tour-TSP is similar to min tour-TSP, except that the problem is to locate a maximum-weight Hamiltonian circuit, with the weight function assumed to be \emph{non-negative} \cite{maxmintsp}.  
Despite being seemingly similar, max tour-TSP turns out to be easier to approximate than min tour-TSP; the currently best known deterministic algorithm runs in cubic time (in the number of vertices) and has approximation ratio 61/81 for an arbitrary non-negative weight function \cite{maxTSP}.

More importantly, we note that max tour-TSP and its variants continue to invite improvements \cite{maxmetrictsp2009},\cite{maxTSP}, and that the weight function does not have to be metric for a polynomial-time approximation algorithm with constant approximation ratio to be devised. 

Unlike min TSP, max TSP approximations are appropriate only for the \emph{tour} versions. For our purposes, we will use the first version of min path-TSP to show NP-completeness of maximum-score 2-ordering in Section \ref{windowing}, while approximate solutions to max tour-TSP will be used for devising approximate solutions to maximum-score 2-ordering in Section \ref{approximation}. 

\section{The w-ordering problem}\label{windowing}
\subsection{Sorted Neighborhood}
To begin, Sorted Neighborhood assumes a blocking key to be given. For clarity, the \emph{functional} definition of a blocking key is provided below:
\begin{mydef}\label{bk}
\emph{Given a set $R$ of records and an alphabet $\Sigma$, define a \emph{blocking key} to be a function $b: R \rightarrow \Sigma^*$}
\end{mydef}  
Let $b(r)$ (for some $r \in R$) be denoted as the blocking key value (BKV) of $r$. Given a finite set of records $R$ and blocking key $b$, let $Y$ be denoted as the \emph{set} of BKVs for $R$. Note that $|Y| \leq |R|$. The inequality is strict if more than one record has the same BKV. 

Assume a \emph{total order} on $\Sigma^*$, and by consequence, $Y$. In keeping with the earlier assumption in Section \ref{relatedwork}-A that processing a record is a \emph{constant-time} operation, BKV computation should not be an expensive operation \cite{stolfo}, \cite{christensurvey}.

Given the run-time of $b$ to be $t(b)$ per record, an SN algorithm would first generate $Y$ in time $O(|R|t(b))$, and then convert $R$ into a sorted list, $R^l$, using the BKVs in $Y$ as the sorting keys. Assuming a comparison sort, the step would take $O(|R| log |R|)$ and was found to be the most expensive SN step in practice \cite{stolfo, christensurvey}. This also implies that $t(b)$ is usually $o(log|R|)$. Henceforth, we consider $t(b)$ to be $O(1)$. In Section \ref{approximation}, we lift this assumption and consider arbitrarily expensive blocking keys when we present and analyze approximation algorithms. 

In the \emph{merge} step, the \emph{w-window} is slid from the first record in $R^l$ to the last record in exactly $|R^l|-w+1$ sliding steps. In each such step, pair the \emph{first} record in the window with all \emph{other} records sharing the window, to add exactly $w-1$ unique pairs to the candidate set $\Gamma$. In the final sliding step, pair \emph{every} record in the window with every other record to generate $w(w-1)/2$ pairs, ensuring that all\footnote{In a simplified implementation, the last window would not be treated differently from the other windows \cite{stolfo}. This will not change the analysis since it only removes a constant additive term} records sharing a window are paired and added to $\Gamma$ \cite{christensurvey}. 

An advantage of SN is that $|\Gamma|$ exactly equals $(|R^l|-w)(w-1)+w(w-1)/2$ and is a \emph{deterministic} function of $w$. It is independent of the blocking key $b$, and the actual distribution of BKVs that $b$ generates.
Referring to Table \ref{example1} again, consider the merge step for $w=3$. There would be $7-3+1=5$ sliding steps. In the first four steps, two unique pairs are generated and added to $\Gamma$. In the final step, three pairs are generated. $\Gamma$ contains $(7-3)*2+3*(3-1)/2=11$ pairs.

Let $\Gamma_m$ be the subset of true positives included in $\Gamma$. The \emph{Pairs Completeness} (PC) of $\Gamma$  is defined as $|\Gamma_m|/|\Gamma|$ \cite{tailor}. The metric is commonly used to evaluate blocking procedures and is an indication of the \emph{coverage} or \emph{recall} of the candidate set \cite{christensurvey}.
As described informally in Section \ref{introduction}, the sorting of $R^l$ can make a difference to $PC$, if it results in true positives getting left out of $\Gamma$. Since $Y$ already has a total order, the problem occurs if records share the \emph{same} BKV. Suppose a set of $q > 1$ records $R_y=\{r_1,\ldots, r_q\}$ have the same BKV $y$. Notationally, the $q$ records are said to fall within the same \emph{block} $R_y$, identified by the BKV $y$ \cite{christensurvey}.

To break ties, an additional input is required, similar to the blocking key $b$, but operating at a finer level of granularity. This motivates us to define a \emph{scoring heuristic} $f$:
\begin{mydef}\label{f}
\emph{Given a set $R$ of records, define a \emph{scoring heuristic} on \emph{unequal} inputs to be a \emph{symmetric} function $f: R \times R \rightarrow \mathbb{R}^+\cup \{0\}$, with run-time per invocation bounded above by $O(|R|^c)$ for some constant $c$. $\forall r \in R$, $f(r,r)$ is undefined.}
\end{mydef}  

Given a set of pairs (for example, the candidate set $\Gamma$), the \emph{score} of that set can be computed by calculating and summing the score of every pair in the set. Given a list of records $R^l$, the score of the list will depend on the window size $w$. Specifically, the merge step will first have to be run on the list and the score, calculated for the generated set of pairs. Algorithm \ref{alg1} summarizes the process. We designate the score of the list returned by Algorithm \ref{alg1} as the \emph{w-score}, since the score depends on $w$. 

A \emph{maximum-score w-ordering} for a \emph{set} $R$ of records can now be defined: 
\begin{mydef}\label{wordering}
\emph{Given a set $R$ of records, a constant window parameter $w$, and a scoring heuristic $f$, define the \emph{maximum-score w-ordering} for $R$ to be an ordering of $R$ given by the list $R^l$, such that a strictly higher w-score exists for no other ordering. }
\end{mydef} 

Intuitively, while (without imposing additional constraints) it is incorrect to think of $f$ as a \emph{probability density} function, a high score on an input record pair indicates a high degree of belief\footnote{On the part of the domain expert who provided$f$ and $b$} that the pair should  be included in $\Gamma$. Although we have not defined $f$ to be dependent on $b$ in any way, a practical design would probably consider both functions in tandem. We further note that even though $f$ is restricted to run in polynomial time, it would, in practice, be expected to be inexpensive, similar to the blocking key $b$. Many such heuristics have been documented in the literature \cite{datamatching}, two good examples being \emph{token-Jaccard} and \emph{cosine similarity}. Both functions are commonly in use in the record linkage community, and are known to work well in a variety of blocking scenarios \cite{datamatching}.

Recall that Sorted Neighborhood first assigns each record a BKV and generates a set of blocks, where each block is a set $R_y$ of records sharing the same BKV $y$. let us assume that a total of $u$ BKVs were generated and that the set $Y$ of BKVs is $\{y_1,\ldots, y_u\}$. We also noted that $|Y|$ (and therefore, $u$) can be at most $|R|$ because of the functional definition\footnote{\emph{Many-many} blocking keys do exist \emph{beyond} the scope of Sorted Neighborhood and are used in some modern blocking methods \cite{christensurvey}; we do not consider them in this paper} of the blocking key in Definition \ref{bk}. Let the total order on $Y$ be $y_1 \leq \ldots \leq y_u$ and the sorting order be ascending. After the BKV generation and sorting phase, we are left with an ordered list of \emph{blocks} $<R_{y_1},\ldots,R_{y_u}>$. 

Before running the merge (or sliding window) step, each block should \emph{ideally} be ordered so that the w-score of the resulting ordered list of \emph{records} $R^l$ is maximized for given $w$ and $f$. This is a \emph{constrained} ordering problem; the maximum-score w-ordering of the full set $R$ of records might yield a list that potentially disobeys the total order imposed on $Y$. In other words, the list could yield a candidate set that is not a \emph{valid} Sorted Neighborhood output, given the inputs. Given this observation, we define a \emph{maximum-score Sorted Neighborhood} (max SN) as follows: 
\begin{mydef}\label{msn}
\emph{Given a scoring heuristic $f$, windowing constant $w$ and blocking key $b$, define \emph{maximum-score Sorted Neighboorhood} (max SN) as a Sorted Neighborhood algorithm that generates (from all valid candidate sets) a candidate set $\Gamma$ with maximum score.}
\end{mydef}  

Since max SN is an SN algorithm, it must obey the ordering constraint just described. Given the various inputs, the candidate set output by max SN is the best result achievable for the SN blocking method. Note that changing any of these inputs (while keeping $R$ intact) can lead max SN to output a different candidate set. 

Furthermore, depending on how the scoring heuristic is defined, a max SN algorithm can be realized in two different ways. First, define a \emph{local} scoring heuristic as a scoring heuristic that is constrained to returning 0 for every record pair $(r,s)$ such that $r$ and $s$ have different blocking key values. On the other hand, a \emph{global} scoring heuristic has no such constraint, except for the ones imposed in the original Definition \ref{f}. 

For local $f$, we can claim the following:
\begin{thm}\label{claim1}\emph{
If $f$ is local and the window size is $w$, maximum-score w-ordering each block \emph{independently} in the ordered list of blocks $<R_{y_1},\ldots,R_{y_u}>$ is both necessary and sufficient for max SN.} 
\end{thm}
\begin{proof}
In Appendix.
\end{proof}
The proof sketch of this theorem is fairly intuitive; the key observation to note in proving the claim is that if a block were \emph{not} maximum-score w-ordered, then such an ordering would yield a strictly higher score for $\Gamma$. Since paired records not from the same block have score 0, such an ordering cannot influence the scores contributed by neighboring blocks. The claim also demonstrates why we designated this particular definition of $f$ as \emph{local}, since each block can be maximum-score w-ordered locally, in order to achieve a global maximum.

We show that, even for $w=2$ and some integer $k'$, merely determining the existence of an ordering for a set $R$ of records, such that the ordering has w-score at least $k'$, is an NP-complete decision problem. We call this problem the \emph{maximum-score w-ordering problem}, since an oracle to the problem can be used to find such an ordering, similar to related oracles for other NP-complete decision problems. 
\begin{algorithm}[t!]
\caption{Compute w-score for record list $R^l$}
\begin{spacing}{0.9}
\begin{algorithmic}
\STATE \textbf{Input :} \begin{itemize}
\item A list $R^l$ of records
\item A windowing constant $w$
\item A scoring heuristic $f$
 \end{itemize}
\STATE \textbf{Output :} \begin{itemize}
\item A real valued w-ordering score \emph{w-score}  
 \end{itemize}
\STATE \textbf{Method :}
\begin{enumerate}
\STATE Initialize empty set of pairs $\Gamma$
\STATE Initialize \emph{w-score} to $0$
\IF{$|R^l| \leq w$}
\STATE $\Gamma= \{\{r, s\} | r \neq s\}$, $r,s$ are records in $R^l$
\STATE Goto line 8
\ENDIF
\FORALL {$i \in \{1, \ldots |R^l|-w+1\}$}
\FORALL {$j \in \{i+1, \ldots i+w-1\}$}
\STATE $\Gamma= \Gamma \cup \{\{R^l[i], R^l[j]\}\}$
\ENDFOR
\ENDFOR
\STATE $\Gamma=\Gamma \cup \{\{R^l[i], R^l[j]\} | i \neq j \wedge i,j \in \{|R^l|-w+2, \ldots |R^l|\}\}$
\FORALL {$\{r,s\} \in \Gamma$}
\STATE \emph{w-score}=\emph{w-score}+$f(r,s)$
\ENDFOR
\STATE{Output \emph{w-score}} 
\end{enumerate}
\end{algorithmic}
\end{spacing}
\label{alg1}
\end{algorithm}

\begin{thm}\label{2optim}\emph{
Maximum-score 2-ordering of a set $R$ of records is NP-complete.}
\end{thm}
\begin{proof}
We show a polynomial-time reduction from min path-TSP, introduced in Section \ref{tsp}. The version used in this proof is that of both endpoints being \emph{unknown}. Recall that the decision version of the problem statement is to determine if a Hamiltonian path with cost \emph{at most} $k$ exists in a given complete, undirected, weighted graph $G=(V,E,W)$. The problem is known to be NP-complete even if weights are \emph{non-negative integers}, as we assume\footnote{It stays NP-complete even if weights are only 0-1 by reducing from the \emph{Hamiltonian path} problem \cite{algorithms}} \cite{hoogeveen}.

We begin the reduction by bijectively mapping each vertex $v \in V$ to a record $r$ and placing all mapped records in a set $R$. Suppose $|V|=m \geq 2$. Then the set $R$ contains the records $\{r_1,\ldots , r_m\}$.
Define the non-negative integer $W_E$ to be $\sum_{e \in E} W(e)$ where $W(e)$ is the weight of edge $e$. $W_E$ is a non-negative integer because all weights were assumed to be non-negative integers. We construct $f$ as a \emph{symmetric look-up table} as follows: the score between any two distinct records $r_i$ and $r_j$ ($i \neq j$, $i,j \in \{1,\ldots , m\}$)  is simply $W_E-W(\{v_i,v_j\})$, where $v_i$ and $v_j$ are the corresponding vertices. Constructed this way, $f$ is both symmetric and non-negative and is therefore an eligible scoring heuristic. $f$ also runs in polynomial-time, since each look-up requires (at worst) a pass over a table occuping $O(|R|^2)$ space.

The entire construction takes quadratic time. We query the maximum-score 2-ordering oracle for the existence of a list with 2-score at least $k'$, with $k'=W_E (m-1)-k+1$ (recall that $m=|V|=|R|$). $k'$ is an integer, since $W_E$ is an integer. We claim that the (boolean) output of the oracle is also the output of the original min path-TSP problem instance; hence, we claim a correct \emph{Karp reduction}. 

First, we prove correctness for \emph{True} oracle outputs. A \emph{True} output implies that a list with 2-score at least $k'$ exists; let the list be $<r_1,\ldots, r_m>$ without loss of generality. The 2-score of this list (per Algorithm \ref{alg1} semantics) is $\sum_{i =1}^{i =m-1} f(r_i,r_{i+1})$. By construction, $f(r_i,r_{i+1})=W_E-W(\{v_i,v_{i+1}\})$ and therefore, $\sum_{i =1}^{i =m-1} f(r_i,r_{i+1})=\sum_{i =1}^{i =m-1}(W_E-W(\{v_i,v_{i+1}\}))$. Because $W_E$ is independent of $i$ and the summation is over $m-1$ elements, we can rewrite the right hand side as $W_E(m-1)+\sum_{i =1}^{i =m-1}W(\{v_i,v_{i+1}\})$. Since the oracle returned \emph{True}, this quantity is at least $k'$; in other words, $W_E(m-1)+\sum_{i =1}^{i =m-1}W(\{v_i,v_{i+1}\}) \geq k'=W_E (m-1)-k+1$ by definition of $k'$ above. This in turn implies that $\sum_{i =1}^{i =m-1}W(\{v_i,v_{i+1}\}) \geq -k+1$ or $k < \sum_{i =1}^{i =m-1}W(\{v_i,v_{i+1}\})+1$. Since $k$ is an integer, this shows that $\sum_{i =1}^{i =m-1}W(\{v_i,v_{i+1}\}) \leq k$. But this implies that a Hamiltonian path\footnote{Graph theoretically, a path is defined as an alternating sequence of vertices and edges \cite{algorithms}} $<v_1, \{v_1,v_2\},v_2,\ldots,\{v_{m-1},v_m\},v_m>$ with weight at most $k$ exists. 

If the oracle returns \emph{False}, we can use the same sequence of equations and a proof by contradiction to show that it cannot be the case that a Hamiltonian path with weight at most $k$ exists in the input graph, since if it did, a corresponding list can be constructed with 2-score at least $k'$. Together, these arguments show that the proposed Karp reduction is valid. 

Finally, to show that the maximum-score 2-ordering problem is in NP, we accept a list $R^l$ as a certificate, input $R^l$ (with $w=2$) to Algorithm \ref{alg1} and compare the 2-score output by the algorithm to the input decision constant $k'$. For constant $w$ and polynomial-time $f$, Algorithm \ref{alg1} runs in (polynomial) time $O(t(f)|R^l|)$. The verification algorithm is therefore polynomial and maximum-score 2-ordering is in NP. 
In combination with the first part of the proof, we conclude that maximum-score 2-ordering is NP-complete.   

\end{proof}

\begin{cor}\label{cor1}\emph{
Maximum-score 2-ordering of a set $R$ of records is NP-complete for a scoring heuristic $f$ of the form $f=1-f'$, where $f'$ is metric with range [0,1].} 
\end{cor}
\begin{proof}
We reduce from min \emph{metric} path-TSP (again with both endpoints unknown), which is also known to be NP-complete \cite{TSPpath}. We construct $f$ to have the form $1-f'$ (instead of subtracting from $W_E$, the sum of all edge weights), where $f'$ is the metric weight function of the input graph $G$. Also, $k'$ in Theorem \ref{2optim} is set to $(m-1)-k+1$. Otherwise, the proof remains virtually the same as that of Theorem \ref{2optim}.    
\end{proof}
The form above is important because many of the similarity functions in the literature adhere to it \cite{datamatching}. Examples include Jaccard and cosine similarities, whose \emph{distance versions} are known to be metric \cite{datamatching}. The corollary shows that, even for this special case, the problem is no easier.

This is also the case as $w$ increases, assuming $w$ is still a constant. Note that, while Theorem \ref{2optim} can be used to prove that maximum-score w-ordering is NP-complete for an \emph{arbitrary} constant $w$, it does not prove NP-completeness for \emph{any} constant $w \geq 2$. To understand the difference, consider the classic NP-complete problem of proving satisfiability of a k-CNF formula, for an arbitrary $k \geq 2$. This problem is NP-complete, since 3-CNF satisfiability is known to be NP-complete. However, this would not be true for \emph{any} constant $k \geq 2$, since 2-CNF satisfiability is known to be in P \cite{algorithms}. 

Reducing the ordering problem from typical variants of min TSP is not straightforward for $w>2$, since we relied on the fact that $w=2$ in the problem reduction. Instead, we reduce the maximum-score 2-ordering problem, which Theorem \ref{2optim} showed to be NP-complete, to the maximum-score w-ordering problem for a constant $w>2$ by using a polynomial-time \emph{scaling} mechanism. The proof is quite technical; we reproduce it in the Appendix for the interested reader.
\begin{thm}\label{coptim}\emph{
Maximum-score w-ordering of a set $R$ of records is NP-complete, for any constant $w > 2$. 
}\end{thm}
\begin{proof}
In Appendix. 
\end{proof}
\begin{cor}\label{cor2}\emph{
Maximum-score w-ordering of a set $R$ of records is NP-complete, for any constant $w \geq 2$. }
\end{cor}
Mirroring the case of Corollary \ref{cor1}, the ordering problem for $w>2$ continues to be NP-complete even for metric $f$. The technical details are not repeated here.

\begin{figure}[t]
\includegraphics[width=7cm, height=4cm]{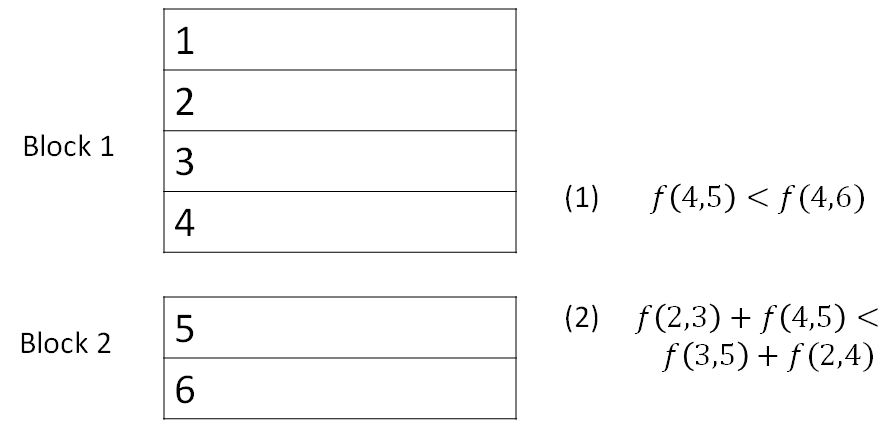}
\caption{A construction showing why maximum-score w-ordering each block separately is neither sufficient (given (1)) nor necessary (given (2)), assuming $w=2$, an arbitrary $f$ and that each block in the figure is maximum-score 2-ordered}
\label{gh}
\end{figure}
Devising a max SN algorithm even for local scoring heuristics becomes an NP-hard problem because of Theorem \ref{coptim}, since each independent block needs to be maximum-score w-ordered (Theorem \ref{claim1}). The problem is no easier if the heuristic is global, since if a polynomial-time oracle exists for an arbitrary heuristic $f$, it can be used to solve the special case of local $f$. 

There are other consequences of having a global scoring heuristic. First, Theorem \ref{claim1} is no longer true. A simple construction in Figure \ref{gh} illustrates why. Consider just two blocks that have been individually maximum-score 2-ordered. By itself, the \emph{first} condition ($f(4,5)<f(4,6)$) in the construction shows that this is no longer sufficient, since reversing records 5 and 6 will still be a maximum-score 2-ordering for Block 2, but the score of the overall candidate set will be higher. By itself, the \emph{second} condition shows that the maximum-score 2-ordering is also not necessary. Intuitively, if the scores are high for some pair of records straddling blocks, then this could theoretically be enough to compensate for any gains that could be achieved from local maximum-score 2-ordering that does \emph{not} place those records at the block \emph{boundaries}, as in the example.

Thus far, we assumed a single blocking key $b$ and \emph{single-pass} Sorted Neighborhood, but the analysis can be extended in a straightforward way to multi-pass Sorted Neighborhood. Specifically, assume a set $B=\{b_1,\ldots,b_c\}$of $c$ blocking keys, where $c \geq 1$ is some constant. For each key $b_i \in B$, the single-pass SN procedure is run and a candidate set $\Gamma_i$ is output. In this way, $c$ independent passes are run, and $c$ candidate sets are output. The final candidate set output by the entire procedure is simply the union of all $c$ sets \cite{stolfo},\cite{stolfo98}.

While the asymptotic analysis does not change, multi-pass SN runs slower (in practice) by a factor of $c$ on a serial architecture. Because passes are independent, both multi-core and shared-nothing parallel architectures are appropriate for the problem \cite{stolfo98},\cite{snmapreduce}. We can define \emph{max multi-pass SN} as a multi-pass SN algorithm with each individual pass meeting the requirement set in Definition \ref{msn}. Note that the windowing constant $w$ is assumed to be fixed over all passes, but each pass can have its own scoring heuristic. Formally, each pass now takes a pair $<b,f>$ as input, with a total of $c$ distinct pairs for a $c$-pass procedure. This further implies that there need not be $c$ distinct blocking keys and scoring heuristics, merely $c$ distinct pairs. 

Multi-pass SN has widely emerged as the method of choice (over single-pass SN) both because of parallelism and also because it was experimentally verified to increase the recall of $\Gamma$, mainly due to using a diverse set of blocking keys \cite{stolfo98}. In combination with a small window constant $w$, inclusion of false positives in the candidate set was also found to be greatly reduced \cite{stolfo}. 

\subsection{Traditional blocking}
Although the primary focus of this work is Sorted Neighborhood, we briefly show how the ordering problem arises in the hash-based \emph{traditional blocking} method, which continues to enjoy popularity due to its simple implementation \cite{christensurvey}. In the original version, a functional blocking key is assumed and each record is assigned a single BKV, exactly like in SN. However, a total order is \emph{not} assumed on the set of BKVs $Y$, which implies that the records cannot be sorted. Instead, each block $R_y$ is treated like a hash bucket, and the hash key is simply the BKV $y$. Records sharing a block are paired and added to $\Gamma$.

This last step leads to problems when we consider the issue of \emph{data skew} \cite{christensurvey}. If some block contains far too many records compared to other blocks, pairing records in that block will dominate run-time. Sorted Neighborhood systematically dealt with the issue by using a constant window size in the merge step. To address the same issue in traditional blocking, several ad-hoc techniques have been proposed \cite{papadakis}.  

In our recent work, we used the same sliding window procedure as SN in \emph{each} block and generated the resulting candidate set \cite{mayankblocking}. We showed competitive empirical performance of the technique (on standard benchmarks) even with a simple token-based blocking key. To distinguish this procedure from the SN merge procedure, we designate it as the \emph{block-merge step}, since the merge step is run on each block in isolation. 

It is straightforward to adapt the formalism presented thus far, if we assume the block-merge procedure for controlling data skew in traditional blocking. Maximum-score traditional blocking (or \emph{max traditional blocking}) can be defined in the same vein as in Definition \ref{msn}. Note that since a total ordering on $Y$ (and hence, \emph{stacking} of blocks against each other) does not exist in traditional blocking, the difference between global and local $f$ does not arise.

Finally, we can state a version of Theorem \ref{claim1} for max traditional blocking.
\begin{thm}\label{otb}\emph{
For any constant $w$, scoring heuristic $f$, and a set $\Pi$ of blocks generated by a functional blocking key, maximum-score w-ordering each block in $\Pi$ individually is both necessary and sufficient for max traditional blocking, assuming the block-merge procedure for generating the candidate set. }
\end{thm}
The proof is quite similar to that of Theorem \ref{claim1}; we do not repeat it. 
\begin{figure}[t]
\includegraphics[width=8cm, height=6cm]{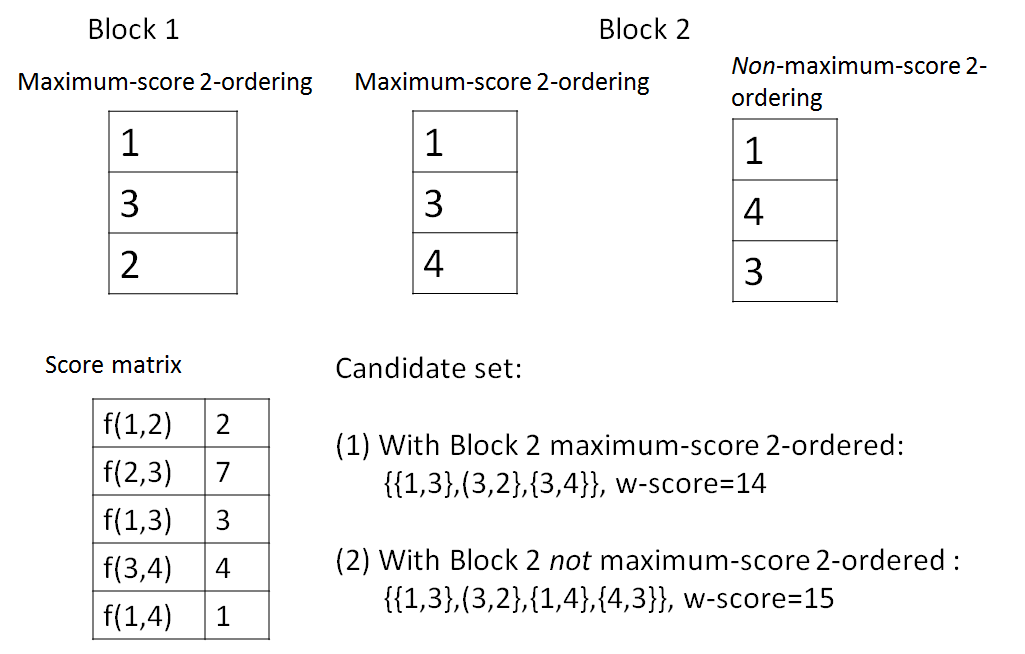}
\caption{A construction showing that Theorem \ref{otb} does not hold for many-many traditional blocking. The maximum-score 2-ordering of each block and the 2-scores may be verified by a brute-force calculation using the provided $f$}
\label{manymany}
\end{figure}

Interestingly, even though the difference between global and local $f$ does not arise for traditional blocking, a related issue arises if we consider traditional blocking with \emph{non-functional} blocking keys. Such keys can assign multiple blocking keys to a record. Just like Theorem \ref{claim1} was shown not to hold for global $f$, a simple construction (Figure \ref{manymany}) shows that Theorem \ref{otb} does not necessarily hold for \emph{many-many} traditional blocking. We leave for future work to determine the appropriate conditions for guaranteeing an optimal candidate set both for many-many traditional blocking, as well as Sorted Neighborhood with a global scoring heuristic.

\section{Approximate Solutions}\label{approximation}
Theorem \ref{2optim} showed that maximum-score 2-ordering is NP-complete. Thus, the next best course of action is to devise polynomial-time approximation algorithms, preferrably with good constant bounds. Given the close connection between the 2-ordering problem and TSP, and the progress in max tour-TSP approximations (Section \ref{tsp}), a natural question to ask is whether maximum-score 2-ordering can be reduced \emph{to} the appropriate max tour-TSP version. We show subsequently that this is feasible, and we utilize this reduction in the approximation algorithms proposed in this section.

In the rest of this section, we explicitly assume $w=2$. We present three approximation algorithms, with two algorithms addressing multi-pass Sorted Neighborhood for local and global scoring heuristics respectively, and one MapReduce algorithm for traditional blocking with the block-merge procedure and with $w=2$. There are a number of reasons why we focus exclusively on the $w=2$ case.

First, the complexity of multi-pass Sorted Neighborhood, even while neglecting the ordering problem, is known to be $O(c(|R|log|R|+w|R|))$, where $c$ is the number of passes, $w$ is the windowing constant and $R$, the input set of records \cite{stolfo}. Even though $c$ and $w$ are constants, they cannot be neglected in practice, as the experiments in the original paper showed \cite{stolfo}. It was found in the experiments that both run-time and the false-positives included in the candidate set increased rapidly with $w$. The conclusion was that, even for a test database with slightly under fourteen thousand records, the recall of $\Gamma$ with $c=3$ achieved a high value at $w=2$ and remained virtually flat thereafter. The recall was higher than with $c=1$ and with $w$ set to high values. We cited this result as a motivation in Section \ref{introduction}, and it is the main reason we focus on maximizing the performance of multi-pass SN for $w=2$.

The second problem with assuming an arbitrary $w$ is that TSP is no longer applicable. A reduction either to or from TSP is not evident for $w > 2$. The approximability status of the problem is also unknown, in the absence of a clear reduction.

For these reasons, we leave devising approximations for $w>2$ for future work.   
In the rest of this section, assume that a polynomial-time $\rho$-approximation algorithm for max tour-TSP is available as a subroutine, MAX TOUR-TSP. We have already cited one such algorithm in the literature \cite{maxTSP} but in general, any appropriate max tour-TSP approximation algorithm may be used. We note that if a randomized subroutine is used, then proposed algorithms also become randomized and approximation ratios are \emph{expected}, rather than guaranteed.      

Finally, the analysis will depend on the distribution of BKVs generated by the blocking keys input to the algorithms. We will conduct the analysis for both the \emph{uniform} distribution as well as the \emph{Zipf} distribution \cite{zipf}. The uniform distribution is the ideal case, since it assumes no data skew and all blocks have the same number of records. The Zipf distribution involves a realistic amount of data skew. For this reason, both distributions were taken into account in a recent survey of blocking methods \cite{christensurvey}. 

To describe the Zipf distribution, let $H_u$  be denoted as the partial harmonic sum for some positive integer $u$:
\begin{equation}\label{harmonic}
H_u=\sum_{i=1}^{u} \frac{1}{i}
\end{equation}
Given a set $R$ of records and a blocking key that assigns BKVs to records according to the Zipf distribution, let $|Y|=u$, that is, $u$ blocks are generated. In descending order by size, the $m^{th}$ block will have size $|R|/(mH_u)$ 

Attribute values in many practical databases have been known to occur with Zipf-like frequency, including US and Chinese firm sizes \cite{usfirm},\cite{chinesefirm}, and more importantly, personal names \cite{personalname}. The analysis assuming the Zipf distribution for blocking key values is therefore expected to match real-world scenarios more closely than the uniform distribution.  
\begin{algorithm}[t!]
\caption{Multi-pass Sorted Neighborhood, local $f$}
\begin{spacing}{0.9}
\begin{algorithmic}
\STATE \textbf{Input :} \begin{itemize}
\item Set $R$ of records
\item Set $C$ containing $c$ blocking key and local scoring heuristic \emph{pairs}
 \end{itemize}
\STATE \textbf{Output :} \begin{itemize}
\item Candidate set of pairs $\Gamma$  
 \end{itemize}
\STATE \textbf{Method :}
\begin{enumerate}
\STATE Initialize empty candidate set $\Gamma$
\FORALL {pairs $ <b,f> \in C$}
\STATE $\Gamma:=\Gamma \cup $Algorithm \ref{approx1}($R,f,b$)
\ENDFOR
\end{enumerate}
\end{algorithmic}
\end{spacing}
\label{approx0}
\end{algorithm}

\begin{algorithm}[t!]
\caption{Single-pass Sorted Neighborhood, local $f$}
\begin{spacing}{0.9}
\begin{algorithmic}
\STATE \textbf{Input :} \begin{itemize}
\item Set $R$ of records
\item Local scoring heuristic $f$
\item Blocking key $b$
 \end{itemize}
\STATE \textbf{Output :} \begin{itemize}
\item Candidate set of pairs $\Gamma$  
 \end{itemize}
\STATE \textbf{Method :}
\begin{enumerate}
\STATE Initialize empty multimap $M$ containing pairs of key-value-\emph{sets}, with the key being a BKV and the value-set, a set of records 
\STATE Initialize empty set of BKVs $Y$, and empty list $Y^l$
\STATE Initialize empty list of records $R^l$
\STATE Initialize empty candidate set $\Gamma$
\FORALL {records $r \in R$}
\STATE Apply $b$ on $r$ to get BKV $b(r)$
\IF {$b(r) \notin $ \emph{keyset(M)}}
\STATE Add pair $<b(r),\{\}>$ to M
\ENDIF 
\STATE Add $r$ to $M[b(r)]$, the value-set associated with $b(r)$
\STATE Add $b(r)$ to $Y$
\ENDFOR
\STATE Sort $Y$ using total order to get list $Y^l$
\FOR {(in-order) $y \in Y^l$}
\STATE Let $G:=RecordsToGraph(M[y],f)$, where $M[y]$ is the value-set associated with key $y$
\STATE Call MAX TOUR-TSP on $G$
\STATE Call \emph{TourToList} on TSP output to get list of records $R_y^l$
\STATE Append $R_y^l$ to $R^l$ 
\ENDFOR
\STATE Run merge procedure on $R^l$ with $w=2$, populate $\Gamma$
\STATE Output $\Gamma$
\end{enumerate}
\end{algorithmic}
\end{spacing}
\label{approx1}
\end{algorithm}

\begin{figure*}[t]
\includegraphics[width=17.5cm, height=5.5cm]{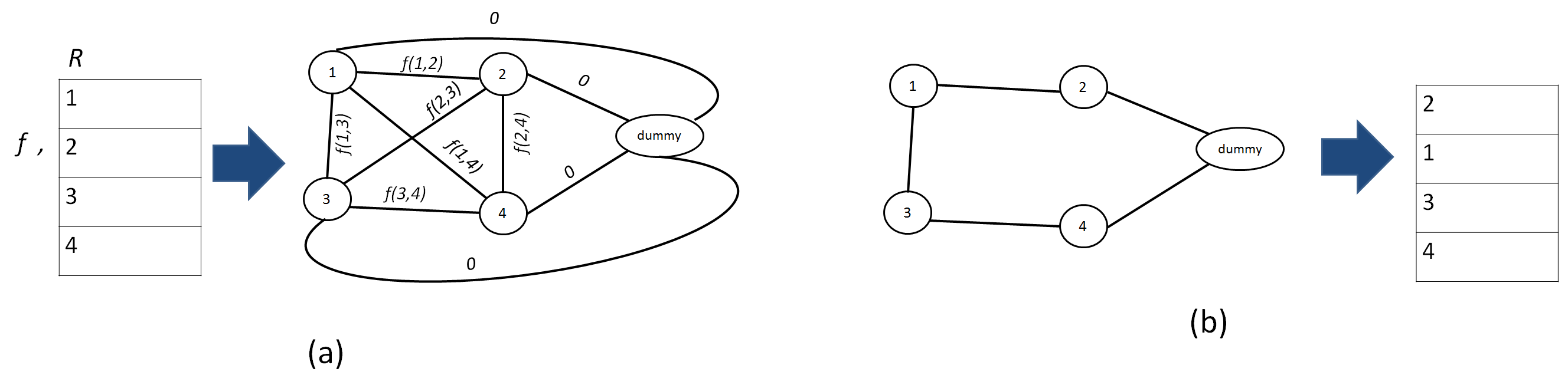}
\caption{The two conversion subroutines used in Algorithm \ref{approx1}. (a) illustrates \emph{RecordsToGraph} and (b) illustrates \emph{TourToList}}
\label{convert}
\end{figure*}

\subsection{Multi-pass Sorted Neighborhood with local scoring heuristics}\label{app1}

Algorithm \ref{approx0} presents the pseudocode for multi-pass Sorted Neighborhood that takes as input a set $R$ of records and a set $C$ of $c$ pairs, with each pair comprising a blocking key and a \emph{local} scoring heuristic. From the discussion at the end of Section \ref{windowing}, this does not imply $c$ unique blocking keys and $c$ unique scoring heuristics. For each of the $c$ pairs, Algorithm \ref{approx0} invokes Algorithm \ref{approx1}, and forms the union of the candidate set output by Algorithm \ref{approx1} and the current candidate set maintained by Algorithm \ref{approx0}. Each iteration of the loop in line 2 is therefore a \emph{pass} in the Sorted Neighborhood sense.    

Algorithm \ref{approx1} presents the pseudocode for approximating a solution to single-pass SN that accounts for 2-ordering, assuming an arbitrary local scoring heuristic. The algorithm begins (lines 1-4) by initializing some data structures, including a multimap with BKVs for keys and with each key pointing to its associated block\footnote{Which is the key's value-\emph{set}; hence, the term \emph{multi}map}. Lines 5-6 perform the BKV computation and block generation step, while line 7 uses the total order to get a sorted list $Y^l$ of BKVs. The list is traversed in order, and for each BKV $y$ in the list, the block $R_y$ is converted into an undirected, complete, weighted graph using the auxiliary subroutine \emph{RecordToGraph}. Figure \ref{convert}(a) illustrates the functionality of the subroutine; we do not provide the technical pseudocode here. Specifically, each record is bijectively mapped to a vertex. In addition a \emph{dummy} vertex is also created. The weight of any edge between two distinct non-dummy vertices $v_1$ and $v_2$ is simply $f(r_1,r_2)$, assuming records $r_1$ and $r_2$ were mapped to vertices $v_1$ and $v_2$ respectively. The weight of any edge between the dummy vertex and any non-dummy vertex is 0. 

MAX TOUR-TSP is then invoked on the graph, and a Hamiltonian circuit is output. A subtle point to note is that MAX TOUR-TSP must work for arbitrary (that is, not necessarily \emph{metric}) non-negative weight functions, regardless of whether the scoring heuristic is metric or non-metric. This is because adding the dummy vertex necessarily makes the weight function non-metric, assuming at least one non-zero weight. Consider an edge $\{v_1,v_2\}$ that has non-zero weight. In the constructed graph, the three edges connecting vertices $v_1,v_2$ and $dummy$ will not satisfy the triangle inequality since the dummy edges are 0. Thus, the 0 sum of the dummy edges will be strictly less than the non-zero weight of the third edge, which is a violation of the triangle inequality. 

This is one motivation for using a max TSP subroutine, since approximation algorithms for max tour-TSP exist that do not place metric assumptions on the weight function \cite{maxTSP}. To leverage better bounds for metric weight functions, a max path-TSP algorithm is required. To the best of our knowledge, none of the max tour-TSP algorithms recently proposed in the literature have been adapted (or can be easily adapted) to solve the path version \cite{maxTSP}. However, it is evident that such an algorithm can be used in Algorithm \ref{approx1}  (in place of MAX TOUR-TSP) if a user so desires, by modifying \emph{RecordsToGraph} so that the extra dummy vertex is not constructed.

\emph{TourToList}, illustrated in Figure \ref{convert}(b), is another straightforward auxiliary subroutine that is invoked on the Hamiltonian circuit output by MAX TOUR-TSP. Construct the list by starting (and ending) the circuit output by TSP from the dummy vertex, after which the dummy vertex is discarded. The list of vertices is reverse mapped to the list of corresponding records. An interesting point is that, in the example in Figure \ref{convert} (b), both (2,1,3,4) and (4,3,1,2) have equal w-scores. This is generally true, given that the graph is undirected. If $f$ is local, the choice of ordering will not matter, and \emph{TourToList} can arbitrarily return one of the two. For global $f$, the choice will matter, an issue we address in the next section. Line 10 runs the sliding window procedure on the generated list $R^l$ and the populated candidate set is output. 

The run-time of Algorithm \ref{approx1} depends on several factors, including the run-time of the blocking key $b$ per record pair and the \emph{distribution} of BKVs generated by $b$. Let the per-invocation run-time of $b$ be denoted as $t(b)$. Assume that $b$ generates $u$ blocking key values. That is, each BKV $y$ is assumed to refer to a block $R_y$ containing an \emph{equal} number ($=|R|/u$) of records. Finally, assume the \emph{amortized} run-time of $f$ to be $t(f)$ per record pair. Amortization is important for some commonly encountered scoring heuristics like \emph{cosine similarity} that have a start-up phase where token statistics (such as frequencies of tokens) need to be collected and stored over the set of records \cite{tfidf}. Once this phase has concluded, computing the similarity takes time that is near-constant, since it only involves a look-up and a simple calculation (like a dot product).   

Regardless of BKV distribution, lines 1-7 in Algorithm \ref{approx1} will run in time $O(t(b)|R|+ulog(u))$. For the remaining analysis, consider first the case of uniform distribution. That is, each of the $hu$ blocks have an equal number of records, which is $|R|/u$. We neglect issues due to rounding for the sake of analysis.   

The time taken by \emph{RecordsToGraph} on a single block is $O((|R|/u)^2t(f))$. Assume the TSP approximation subroutines to run in time $O(|V|^q)=O((|R|/u)^q)$ where $q$ is a constant and is denoted as the \emph{TSP constant}. Typically, $q \leq 3$ \cite{maxTSP}. The merge step (line 10) takes time $\Theta(|R|)$, given it involves exactly $|R|-2+1=|R|-1$ sliding steps. Thus, the total run-time of Algorithm \ref{approx1} for a uniform blocking key is $O(u((|R|/u)^2t(f)+(|R|/u)^q)+(t(b)+1)|R|+ulog(u))$ and the generated candidate set has size exactly $|R|-1$ or $\Theta(|R|)$, since each sliding step generates exactly one pair. Since $u=O(|R|)$ due to the functional definition of a blocking key, the expression above may be simplified further as $O(u((|R|/u)^2t(f)+(|R|/u)^q)+(t(b)+1)|R|)$.

If we conduct the same run-time analysis for Zipf distribution, the complexity of merge (and the candidate set size) remains the same, which we noted at the beginning of Section \ref{windowing} as an advantage of Sorted Neighborhood. Specifically, the size of $\Gamma$ is never dependent on the blocking key. However, the total run-time of Algorithm \ref{approx1} will change, since the TSP subroutine takes as input the graph representation of a single block in each invocation. In accordance with the Zipf distribution, the run-time of Algorithm \ref{approx1} becomes $O(\sum_{i=1}^{u}((|R|/(iH_u))^2t(f)+(|R|/(iH_u))^q)+(t(b)+1)|R|+ulog(u))$=$O(\sum_{i=1}^{u}((|R|/(iH_u))^2t(f)+(|R|/(iH_u))^q)+(t(b)+1)|R|)$. We can rewrite the first two terms as $t(f)(|R|/H_u)^2\sum_{i=1}^{u}1/i^2+(|R|/H_u)^q\sum_{i=1}^{u}1/i^q$. Unfortunately, the summations do not have closed forms, and we cannot simplify the expression further. 

For a comparison to the run-time of the same Sorted Neighborhood procedure that neglects the ordering problem,     

As for qualitative performance, we can prove the following about Algorithm \ref{approx1}:
\begin{thm}\label{approx1ratio}\emph{
The approximation ratio of Algorithm \ref{approx1} is exactly the approximation ratio of MAX TOUR-TSP.}  
\end{thm}
\begin{proof}
In Appendix.
\end{proof}
The run-time of the multi-pass procedure in Algorithm \ref{approx0} is $c$ times the run-time of Algorithm \ref{approx1}, if all the blocking keys have the same BKV distirbutions. This is unlikely, since one of the strengths of the multi-pass procedure is to accommodate diverse blocking keys. Nevertheless, assuming that BKVs generated by the blocking keys in the pairs in $C$ individually obey either the uniform or Zipf distribution, the run-time of Algorithm \ref{approx0} is simply a weighted sum (with weights adding up to $c$) of the appropriate Algorithm \ref{approx1} run-times. If other distributions (not considered in this paper) are accommodated, the analysis can be extended in a straightforward manner. We can also state the following corollary:
\begin{cor}\label{multipassapprox}\emph{
The approximation ratio of Algorithm \ref{approx0} is exactly the approximation ratio of MAX TOUR-TSP.}
\end{cor}
The proof of the corollary is self-evident when we consider the pseudocode of Algorithm \ref{approx0} and Theorem \ref{approx1ratio}.
\begin{algorithm}[t!]
\caption{MapReduce-based Traditional Blocking}
\begin{spacing}{0.9}
\begin{algorithmic}
\STATE \textbf{MAP:} 
\STATE \textbf{Input :} \begin{itemize}
\item Set $R$ of records
\item Blocking key $b$
 \end{itemize}
\STATE \textbf{Output :} \begin{itemize}
\item Set of key-value pairs of form $<b(r),r>$ where $b(r)$ is the BKV of record $r$  
 \end{itemize}
\STATE \textbf{Method :}
\begin{enumerate}
\FORALL {$r \in R$}
\STATE Emit $<b(r),r>$
\ENDFOR
\end{enumerate}
\STATE \textbf{REDUCE:} 
\STATE \textbf{Input :} \begin{itemize}
\item Key-value \emph{set} of form $<y,R_y>$ with the set $R_y$ contains exactly those records with BKV $y$
\item Scoring heuristic $f$
 \end{itemize}
\STATE \textbf{Output :} \begin{itemize}
\item Set of key-value pairs of form $<s,q>$ where $s$ is a double-valued score of unordered record pair $q=\{r,s\}$  
 \end{itemize}
\STATE \textbf{Method :}
\begin{enumerate}
\STATE $G:=RecordsToGraph(R_y,f)$
\STATE Call MAX TOUR-TSP on $G$
\STATE Call \emph{TourToList} on TSP output to get list of records $R_y^l$
\STATE Run block-merge procedure on $R^l$ with $w=2$
\FORALL {pairs $\{r,s\}$ output by block-merge procedure}
\STATE Emit $<f(r,s),\{r,s\}>$
\ENDFOR
\end{enumerate}
\end{algorithmic}
\end{spacing}
\label{approx2}
\end{algorithm}

\subsection{MapReduce-based Traditional Blocking}\label{app2}
The pseudocode for MapReduce-based traditional blocking is shown in Algorithm \ref{approx2}. In the mapper, the blocking key $b$ is applied on each record and the key-value pair $<b(r),r>$ is emitted. This implies that, after the shuffling step, all records with the same BKV end up in the same reducer. Thus, each reducer \emph{instance} processes a single block. Steps 1-5 in each reducer instance mirror the for loop in line 8 of Algorithm \ref{approx1}. Finally, the block-merge procedure, which we described earlier as sliding a constant-sized window $w$ over the records in the block and generating the candidate set thereof, is run on each list output by the TSP subroutine. As in the rest of this section, $w=2$. 

Given that Algorithm \ref{approx2} is designed for traditional blocking and not Sorted Neighborhood, it does not make any difference whether $f$ is local or global. Note that the proof of Theorem \ref{approx1ratio} can be used, with trivial modifications, to prove that the approximation ratio for Algorithm \ref{approx2} is exactly that of MAX TOUR-TSP. 
Finally, the reducer emits key-value pairs of form $<f(r,s),\{r,s\}>$. 

A rigorous analysis of Algorithm \ref{approx2} is not possible without making some assumptions about the available number of reducer instances, as well as the load-balancing strategies of the namenode\footnote{This is the master node that dynamically controls the MapReduce workflow \cite{mapreduce}}. For the sake of analysis, assume that the set $R$ of records is sharded equally on $h_m$ map nodes. Each mapper instance then takes time $O(t(b)|R|/h_m)$. Using notation from the previous sub-section, assume that $u$ distinct BKVs (and hence, blocks) are generated.

If we now assume that at least $u$ reducer instances are available, and the namenode distributes the load equally, the run-time of the reduce phase will be dominated by the \emph{largest} block generated, since this will be the last reducer to terminate. If a uniform distribution on BKVs is assumed, all blocks are of equal size and contain $|R|/u$ records. Using the analysis from the previous sub-section, the total time taken by Algorithm \ref{approx2} is then $O(t(b)|R|/h_m+(|R|/u)^2t(f)+(|R|/u)^q+|R|/u)$. The last term is due to the block-merge procedure and is asymptotically subsumed. The final expression is $O(t(b)|R|/h_m+(|R|/u)^2t(f)+(|R|/u)^q)$.

Assuming Zipf distribution, the largest block size is $|R|/H_u$ with $H_u$ defined earlier as the partial u-term harmonic sum. The run-time of Algorithm \ref{approx2} with Zipf distribution of BKVs is $O(t(b)|R|/h_m+(|R|/H_u)^2t(f)+(|R|/H_u)^q+|R|/H_u)=O(t(b)|R|/h_m+(|R|/H_u)^2t(f)+(|R|/H_u)^q)$.  

\subsection{Multi-pass Sorted Neighborhood with global scoring heuristics}\label{app3}
Unfortunately, there is no evidence yet that an approximation algorithm with constant approximation ratio exists for Sorted Neighborhood with a global scoring heuristic. Given the similarity of the problem to generalized TSP \cite{setTSP}, it could be the case that no such algorithm exists unless $P=NP$. We pose this as a conjecture in Section \ref{futurework}. 

One possible option is to use Algorithm \ref{approx1}, but ignore the fact that $f$ is not local. Theorem \ref{approx1ratio} would no longer apply, but it is reasonable to assume the algorithm will still perform well empirically. The question then is if we can optimize the algorithm further, given that we know $f$ is global and not local.  

Algorithm \ref{approx3} shows the pseudocode for a single-pass Sorted Neighborhood procedure that attempts two optimizations. The algorithm assumes a global scoring heuristic. Note that the only difference in the multi-pass procedure in Algorithm \ref{approx0} for the case of global $f$ is that it would invoke Algorithm \ref{approx3} instead of Algorithm \ref{approx1} in line 2. 

The first optimization is that a \emph{modified} version of Algorithm \ref{approx1} is run to obtain the final list $R^l$. The modification relates to the list polarities of each of the lists returned by the \emph{TourToList} subroutine. Recall that the subroutine has two list choices (denoted as \emph{polarities}) for each Hamiltonian circuit output by MAX TOUR-TSP. The list polarity did not matter if $f$ was local, but because a global $f$ implies that record pairs straddling blocks can have non-zero scores, the polarity of each list matters. For the first block, let \emph{TourToList} randomly return one of the two choices. Next, assuming $u$ blocks, let the $i^th$ block ($i$ ranging from 1 to $u-1$) have record $r$ at the end. Let the $i+1^{th}$ block have \emph{endpoint}\footnote{Technically, the records corresponding to the vertices preceding and following the dummy vertex in the returned Hamiltonian circuit} records $s_1$ and $s_2$. If $f(s_1,r)>f(s_2,r)$, \emph{TourToList} returns the list $(s_1,\ldots,s_2)$, otherwise it returns the reversed list.

With this modification in place, Algorithm \ref{approx1} is run from lines 1-9, and a second optimization called \emph{greedy adjacent swapping} is conducted on the list $R^l$.  To understand the optimization, consider the $i^{th}$ block and the $i+1^{th}$ block ($i$ ranging from 1 to $u-1$), and let the \emph{first} record of the $i+1^{th}$ block be $s$ and the \emph{last} record of the  $i^{th}$ block be $r$. Let the record $r'$ in the $i^{th}$ block have highest score (according to scoring heuristic $f$) when paired with $s$, compared to all other records in the $i^{th}$ block. If $r \neq r'$, swap $r$ and $r'$ if the resulting increase in score ($f(r,s')-f(r,s)$) due to the swap is greater than the (possible\footnote{Because of the approximate nature of the solutions, there is always a small chance that the swap will end up \emph{increasing} the 2-score of that block, which gives us all the more reason to perform the swap}) loss in the local 2-score of the $i^{th}$ block. In the forward pass, $i$ ranges from 1 to $u-1$ (the first to the penultimate block). The backward pass is similar but starts from Block $u$ and goes traverses upwards through the list to Block 2. Each of these two passes yields two different lists $R_f^l$ and $R_b^l$ with their own w-scores. Of the three lists, $R^l$, $R_f^l$ and $R_b^l$, the list with the highest 2-score is output (line 7).

The reason why all three lists must be compared is because greedy adjacent swapping can theoretically lead to a \emph{decline} in the original w-score. Figure \ref{decline} in the Appendix proves this through a construction. 

With uniform BKV distribution, each swap takes time $O(|R|/u)$ since $|R|/u$ records in block $i$ must be compared with the first record in the next block (assuming forward pass). Since $i$ ranges from 1 to $u-1$, the forward pass takes time $O(t(f)|R|(u-1)/u)$ over the time taken by Algorithm \ref{approx1}. Determining list polarity takes time $O(t(f))$ since only two comparisons are required; we assume that it is subsumed by the swapping procedure. 
Similarly, an upper bound of $O(t(f)|R|(u-1)/H_u)$ should be added to the run-time of Algorithm \ref{approx1}, if a Zipf distribution of BKVs is assumed. 

We note that Algorithm \ref{approx3} is amenable to numerous practical optimizations, including caching of scores\footnote{Since it is quite conceivable that many pairs will end up getting scored more than once, as the algorithm evaluates greedy adjacent swapping} and a multi-threaded implementation for each of the forward and backward passes. We leave investigating and evaluating such optimizations for future work.
\begin{algorithm}[t!]
\caption{Single-pass Sorted Neighborhood, global $f$}
\begin{spacing}{0.9}
\begin{algorithmic}
\STATE \textbf{Input :} \begin{itemize}
\item Set $R$ of records
\item Global scoring heuristic $f$
\item Blocking key $b$
 \end{itemize}
\STATE \textbf{Output :} \begin{itemize}
\item Candidate set of pairs $\Gamma$  
 \end{itemize}
\STATE \textbf{Method :}
\begin{enumerate}
\STATE Run \emph{modified} Algorithm \ref{alg1} from lines 1 to 9 with same inputs, to get list $R^l$ 
\STATE Run merge procedure on $R^l$ to get $\Gamma_1$ with score $F_1$
\STATE Perform \emph{greedy adjacent swapping} in a forward pass on $R^l$ to get $R_f^l$
\STATE Run merge procedure on $R_f^l$ to get $\Gamma_2$ with score $F_2$
\STATE Perform \emph{greedy adjacent swapping} in a backward pass on $R^l$ to get $R_b^l$
\STATE Run merge procedure on $R_b^l$ to get $\Gamma_3$ with score $F_3$
\STATE Output as $\Gamma$ the highest-scoring of $\Gamma_1$, $\Gamma_2$ and $\Gamma_3$ with ties broken in that order
\end{enumerate}
\end{algorithmic}
\end{spacing}
\label{approx3}
\end{algorithm}
\subsection{Practical Usage}
Table \ref{algos} lists the three proposed algorithms with run-times for both uniform and Zipf distributions. As noted earlier, the run-time of the multi-pass procedure may be calculated by weighting, if every blocking key either follows a uniform or Zipf distribution. If another distribution is expected, a similar analysis would first have to be carried out for Algorithm \ref{approx1} and the weighting procedure extended appropriately. Finally, note that although we implicitly assume that I/O or shuffling (in the case of Algorithm \ref{approx2}) costs are subsumed in the derived $O$ expression, these costs could be prohibitive for specific databases or implementations and must be separately derived, if this is the case. In the original papers, I/O costs of Sorted Neighborhood were experimentally found not to dominate \cite{stolfo}, \cite{stolfo98}.      

In the context of the full record linkage procedure, a blocking method in Table \ref{algos} should only be selected if its run-time plus $O(t(g)|R|)$ has a strict upper bound $o(t(g)|R|^2)$ where $g$ is the sophisticated similarity function used in the second step. In the last decade, with machine learning procedures, genetic algorithms and expressive feature spaces dominating the state-of-the-art \cite{marlin}, \cite{datamatching}, $g$ has become increasingly expensive. We hypothesize that, with efficient, practical implementations of the proposed algorithms and careful selection of blocking keys and heuristics, the methods will prove to be qualitatively and computationally viable. As an additional advantage, improvements in max TSP will contribute directly to the quality of the proposed algorithms.
\begin{table*}
\caption{A summary of the proposed algorithms. All algorithms accept as input a set $R$ of records, a scoring heuristic $f$ and a blocking key $b$, and output a candidate set $\Gamma$ of size $O(|R|)$. $t$ denotes the run-time (per invocation) of its functional input, $u$ is the number of blocks generated by $b$, $H_u$ is the u-term partial harmonic sum, $h_m$ is the number of map nodes and $q$ is the TSP constant}
    \begin{tabular}[t]{ |p{1.7cm}| p{7.2cm} | p{7.7cm} |}
    \hline
    {\bf Section/Alg.} & {\bf Run-time: Uniform BKV distr.} & {\bf Run-time: Zipf BKV distr.}\\ \hline
   \ref{app1}/\ref{approx1} &$O(u((|R|/u)^2t(f)+(|R|/u)^q)+(t(b)+1)|R|)$ & $O(t(f)(|R|/H_u)^2\sum_{i=1}^{u}1/i^2+(|R|/H_u)^q\sum_{i=1}^{uh}1/i^q+(t(b)+1)|R|)$\\ \hline
 \ref{app2}/\ref{approx2} & $O(t(b)|R|/h_m+(|R|/u)^2t(f)+(|R|/u)^q)$ & $O(t(b)|R|/h_m+(|R|/H_u)^2t(f)+(|R|/H_u)^q)$\\ \hline
 \ref{app3}/\ref{approx3} &Algorithm \ref{approx1}+$O(t(f)|R|(u-1)/u)$ & Algorithm \ref{approx1}+$O(t(f)|R|(u-1)/H_u)$\\ \hline
    \end{tabular}
\label{algos}
\end{table*}
\section{Conjectures and Conclusion}\label{futurework}
This paper shows that devising a maximum-performing Sorted Neighborhood algorithm entails solving an NP-complete w-ordering problem. There is a close connection between 2-ordering and TSP. This connection is used to  define and analyze three approximation algorithms for the special but practically important case of $w=2$. In the future, we will implement and evaluate these algorithms and attempt experimentally viable solutions for $w>2$. We state the following \emph{conjectures}:\begin{itemize}
\item For an arbitrary global heuristic $f$, no polynomial-time $\rho$-approximation algorithm can exist for approximating max SN unless $P=NP$. 
\item For $w>2$ and an arbitrary local $f$, any polynomial-time (in $|R|$) $\rho$-approximation algorithm is exponential in $w-1$.
\end{itemize}
Proving (or disproving) these conjectures will have direct ramifications on the \emph{approximability} of the generic w-ordering problem.
\section*{Acknowledgements}
The authors thank Vijaya Ramachandran and Matthew Johnson for their guidance on Theorems \ref{2optim} and \ref{coptim} respectively.
\bibliography{pods}
\bibliographystyle{IEEEtran}
\newpage
\section*{Appendix}
\subsection{Proof of Theorem \ref{claim1}}
Recall that there is a list of $u$ blocks, $<R_{y_1},\ldots,R_{y_u}>$, where each block is a set of records. We need to show that if each block is maximum-score w-ordered and the scoring heuristic $f$ is local, then the ordered list of records is both necessary and sufficient for max SN. 

From Section \ref{windowing}, the \emph{w-score} of any such list is the \emph{score} of the candidate set $\Gamma$ generated after the list is subject to the w-window merge step. Per Definition \ref{msn}, max SN will always perform merge on a list that guarantees maximum w-score, compared to any other list that obeys SN semantics. Let such a list, $R^l$, be denoted as the \emph{max} list. We assume in this proof that the total number of records in $R$ is strictly greater than $w$, since otherwise, every list would yield the same w-score and max SN would be trivially achieved.

When the merge step commences on the \emph{max} list, all records within a window of size $w$ are paired and added to the candidate set $\Gamma$. There is an alternate way of characterizing $\Gamma$. Specifically, \emph{partition} $\Gamma$ into a set of (at most) $2u$ sets, two for each block. Let the $i^{th}$ block contribute two disjoint sets $\Gamma_i$ and $\Gamma_i^0$. The reason for the 0 superscript will become clear shortly.
 
The partition is achieved by the following construction:
\begin{enumerate}
\item If both records in a pair $\{r,s\} \in \Gamma$ are from the same block $R_{y_i}$, place $\{r,s\}$ in set $\Gamma_i$.
\item If the records in a pair $\{r,s\} \in \Gamma$ are from different blocks (say $r \in R_{y_i}$ and $s \in R_{y_j}$), then add $\{r,s\}$ to $\Gamma_i^0$ if $i < j$ otherwise add the pair to $\Gamma_j^0$. 
\end{enumerate} 

 Some of the sets may be empty; we remove them from the partition if so, and thereby fulfill the conditions of a partition. Using the partition:
\begin{equation}\label{gamma}
\Gamma=\bigcup_{i=1}^u \Gamma_i \cup \bigcup_{i=1}^u \Gamma_i^0
\end{equation}
By Definition \ref{f} of the local scoring heuristic, the score of each pair in $\Gamma_i^0$ is 0 (hence, the superscript). Thus, the score of $\Gamma$ is only contributed to by the first term in Equation \ref{gamma}. Since each of the sets operated upon by union is disjoint, the following equation is directly derived:
\begin{equation}\label{scoregamma}
score(\Gamma)=\sum_{i=1}^u score(\Gamma_i)
\end{equation}
Again, because of disjointness, taking the max on both sides allows us to move the max inside the summation in Equation \ref{scoregamma}. By the semantics of Algorithm \ref{alg1} and Definition \ref{wordering} of maximum-score w-ordering, \emph{max} $score(\Gamma_i)$ exactly equals the maximum w-score achievable over all orderings of the block $R_{y_i}$. Thus, the condition that every block is independently maximum-score w-ordered is a \emph{sufficient} one for maximizing the score of $\Gamma$. 

We prove necessity by contradiction. Suppose some block $R_{y_i}$ is not maximum-score w-ordered. Let its list version (of which the w-score will not be maximized per Algorithm \ref{alg1}) be $R^l_{y_i}$. Let the candidate set generated as a result be $\Gamma'$, and assume $\Gamma$ to be the generated candidate set if $R_{y_i}$ were maximum-score w-ordered. We assume $\Gamma'$ was generated from a max list and show that this leads to a contradiction.
Using Equation \ref{scoregamma} and the construction of the partition, the difference between the scores of $\Gamma$ and $\Gamma'$ will be $score(\Gamma_i)-score(\Gamma'_i)$. Since $R_{y_i}$ was maximum-score w-ordered in the list that generated $\Gamma$ but not $\Gamma'$, the difference is positive. If we replace list $R^l_{y_i}$ with another list that achieves higher w-score, the score of $\Gamma'$ will monotonically improve. This proves the contradiction, since $\Gamma'$ was clearly not generated from a max list.

Thus, if any block is non-maximum-score w-ordered, the resulting global list of records will not be a max list. The \emph{contrapositive} of the statement proves necessity. Coupled with the proof for sufficience, the full statement of the theorem is proved.  

\subsection{Proof of Theorem \ref{coptim}}
Theorem \ref{2optim} showed that maximum-score 2-ordering was NP-complete; hence, it suffices to show a Karp reduction from the decision version of maximum-score 2-ordering on a set $R$ of records. Let $|R|=m$; $R$ consists of records $\{r_1,\ldots,r_m\}$. We also assume that $m>>w$.
 
We begin by constructing a new set $S$ of $m(w-1)$ records, by mapping each record $r_i \in R$ to a set $S_i$ which contains $w-1$ records. Let $S_i$ be denoted as an \emph{internal set} and $i$ as the \emph{set ID} of $S_i$. Also, let each record in this set be denoted as $s_i^j$, where $j \in \{1, 2,\ldots, w-1\}$ and is the \emph{internal ID} of the record. $S$ is simply the union of all the $m$ internal sets that the records were mapped to.

Technically, such a construction can be achieved by assigning $S$ a schema containing just two attributes (called set ID and internal ID). Each record in $S$ can be uniquely identified by employing both IDs in conjunction; hence, both IDs together constitute a compound primary key. 

We obtain a new scoring heuristic $f'$ for record pairs in $S$ using the following construction: 
\begin{enumerate}
\item $f'(s_i^c,s_i^d)=\frac{1}{2}\sum_{u=1}^{m}\sum_{v=1}^m f(r_u,r_v)$, $\forall i \in \{1,\ldots,m\}$ and $\forall c,d \in \{1,\ldots,w-1\}$
\item $f'(s_i^c,s_j^c)=f(r_i,r_j)$, $\forall  i,j \in \{1,\ldots,m\}, i \neq j$ and $\forall c \in \{1,\ldots,w-1\}$
\item $f'=0$ for all other record pairs 
\end{enumerate}

Intuitively, Rule 1 uniformly assigns the same large score to the record pairs that lie within the same internal set, independent of the specific internal or set IDs of records\footnote{This is evident when we consider that the variables on the left hand side of Rule 1 are independent of the variables on the right hand side}. If we visualize the scoring heuristic $f$ as a discrete $m \times m$ matrix (but with diagonal entries undefined; see Definition \ref{f}), the quantity on the right hand side is simply the sum of all (defined) matrix entries. The factor of 1/2 is to prevent each entry from being counted twice, due to the symmetry of the two summations.

Rule 2 shows that a non-zero score can only exist between the $c^{th}$ records of two different internal sets (with $c$ ranging from 1 to $w-1$) and is equal to the score between the original records that were mapped to the two sets in the construction. Note that the score itself is independent of the value of $c$. Rule 3 assigns every other record pair in the scaled problem a score of zero. 

Thus, Rule 1 applies when two records belong to the same internal set (and have the same set ID), while Rule 2 applies when two records have different set IDs but the same internal ID. Otherwise, Rule 3 applies. 

The computation of $f'$ occurs in polynomial time since $w$ is a constant and computing pairwise scores for $m(w-1)$ records is at most $O(t(f)m^2(w-1)^2)$. By definition, $t(f)$ is polynomial in $m$ and $w-1$ is constant. 

Finally, recall that the decision version of maximum-score 2-ordering accepted as input a \emph{decision constant} $k$. A valid solution must return \emph{True} if a list with 2-score at least $k$ exists, where $k$ is a positive integer. Construct the decision constant $k'$ of the transformed problem instance using the equation below:
\begin{equation}\label{k'}
k'=T_1+T_2
\end{equation}
Here, $T_1$ and $T_2$ are given by the equations below:
\begin{equation}\label{t1}
T_1=m{w-1 \choose 2}\frac{1}{2}\sum_{u=1}^{m}\sum_{v=1}^m f(r_u,r_v)
\end{equation}
\begin{equation}\label{t2}
T_2=k(w-1)
\end{equation}
We perform maximum-score w-ordering on this transformed problem instance, with the inputs being the constructed set of records $S$, the scoring heuristic $f'$ and the decision constant $k'$. We claim that the (\emph{True} or \emph{False}) output of the oracle solving the transformed problem is exactly the output of the original problem instance. In other words, we claim a correct Karp reduction.

We start by showing the correctness of the Karp reduction for the \emph{False} output. We prove by contradiction that if the oracle return \emph{False}, then it cannot be the case that an ordering of the original set $R$ of records exists, with 2-score at least $k$, assuming the score heuristic $f$. 

Suppose not. That is, the oracle returned \emph{False} but there is some ordering $R^l$ that has 2-score at least $k$. Without loss of generality, let $R^l=<r_1,\ldots,r_m>$. Since this list has 2-score at least $k$, the following is true: 
\begin{equation}\label{kclaim}
\sum_{i=1}^{m-1} f(r_i,r_{i+1}) \geq k. 
\end{equation}
Given this information, consider (in the transformed problem) the list $S^l$ formed by concatenating the \emph{sub-lists} $S_1^l,\ldots , S_m^l$, where each $S_i^l$ is simply $<s_i^1,\ldots ,s_i^{w-1}>$ for all $i$ ranging over set IDs (that is, from 1 to $m$). Let us calculate the w-score of this list. 

Since the window has size $w$ while a sub-list has $w-1$ records, it will be the case that each sub-list $S^l$ will fall entirely within some window, and therefore, all records within a sub-list will be paired and added to the (transformed problem) candidate set $\Gamma'$. Therefore, each sub-list, by itself, contributes exactly ${w-1 \choose 2}$ pairs. Given Rule 1 in the construction of $f'$ and that there are exactly $m$ sub-lists, the expression $T_1$ (Equation \ref{t1}) will be the score of all record pairs in $\Gamma'$ such that both records in the pair are from the same sub-list.
\begin{figure}[t]
\includegraphics[width=8cm, height=2cm]{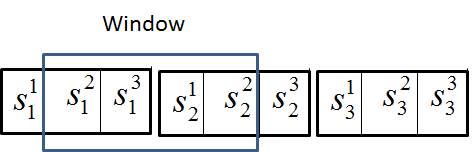}
\caption{An illustration of three concatenated sub-lists, with $w=4$. Looking at the window, the only record pair from different internal sets that can contribute a non-zero score is $\{s^2_1,s^2_2\}$}
\label{alignment}
\end{figure}

Also, because of the way each sub-list is defined, it must be the case that in \emph{every} window, exactly one pair of the form $\{s_i^c,s_{i+1}^c\}$ can contribute a non-zero score, among record pairs where both records are from different sub-lists. Consider Figure \ref{alignment} for an illustration, with $m=3$ and $w=4$. Since $c$ ranges from 1 to $w-1$ and $i$ ranges from 1 to $m-1$, such pairs will contribute total score (using Rule 2):
\begin{equation}\label{t'2}
T'_2=(w-1) \sum_{i=1}^{m-1} f(r_i,r_{i+1})
\end{equation}
By Equations \ref{t2}, \ref{kclaim} and \ref{t'2}, $T'_2 \geq T_2$. But this implies that the w-score of the constructed list $S^l \geq T_1 +T_2$, which implies that the w-score is at least $k'$, by Equation \ref{k'}. This leads to a contradiction, since the oracle returned \emph{False}. Thus, it must be the case that if a list with 2-score at least $k$ exists in the original instance, a list with 2-score at least $k'$ exists in the transformed instance. By the contrapositive, the reduction is correct, assuming the oracle returns \emph{False}. 
\begin{figure}[t]
\includegraphics[width=7cm, height=1cm]{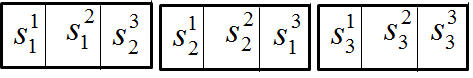}
\caption{An illustration of an interleaved list. Intuitively, the list is interleaved because records $s^3_2$ and $s^3_1$ are not `lined up' with other records from their respective internal sets}
\label{interleaved}
\end{figure}

In order to show the correctness of the oracle for a \emph{True} output, we introduce some additional terminology. Define a list $S^l$ to be an \emph{interleaved} list if there exist distinct records $s_i^c$ and $s_i^d$ in the list, such that the list contains, between these records, at least one record of the form $s_j^e$ where $i,j \in \{1,\ldots ,m\}$, $c,d,e \in \{1,\ldots ,w-1\}$ and $i \neq j$. Let a list that is not \emph{interleaved} be defined as a \emph{stacked} list. For example, the list in Figure \ref{alignment} is stacked. The list in Figure \ref{interleaved} is interleaved. 

Intuitively, a list is interleaved if records from some internal set $S_i$ are not all \emph{lined up} against one another. Place each of $|S|!$ orderings of $S$ into either the set of interleaved orderings $\mathcal{I}^l$ or of stacked orderings $\mathcal{S}^l$. Together, $\mathcal{I}^l$ and $\mathcal{S}^l$ form a partition of the set of all orderings. Note that a stacked ordering is simply the concatenation of sub-lists, where each sub-list is of form $S^l_i$, one of $(w-1)!$ possible orderings of internal set $S_i$. 

\begin{figure}[t]
\includegraphics[width=7cm, height=1cm]{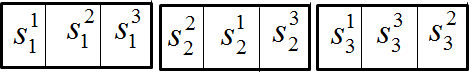}
\caption{An example of a list that is stacked but unaligned, since the order of internal IDs in the first sub-list is 1,2,3 but in the second sub-list is 2,1,3 (and in the third, 1,3,2). The list would be unaligned even if just one of its sub-lists were `out of alignment'}
\label{unaligned}
\end{figure}
Furthermore, define a stacked ordering to be \emph{aligned} if the internal ID of the $c^{th}$ record in every sub-list (with $c$ ranging from 1 to $w-1$) is the same. If a stacked ordering is not aligned, let it be denoted as unaligned. Figure \ref{alignment} is an example of a stacked aligned ordering, while Figure \ref{unaligned} is an example of a stacked unaligned ordering. The set $\mathcal{S}^l$ is then partitioned into the sets $\mathcal{S}_a^l$ and $\mathcal{S}_u^l$ of stacked aligned and unaligned stacked orderings respectively.

Define the alignment function to be a function that takes a stacked unaligned ordering as input and aligns it by using a particular sub-list as a \emph{pivot}. Thus, the output is a stacked aligned ordering. For example, Figure \ref{alignment} is the alignment of Figure \ref{unaligned} if we use the \emph{first} sub-list as the pivot. Specifically, the function \emph{rearranges} the records in every sub-list except the pivot sub-list, such that the order of internal IDs in every sub-list now reflects the order of internal IDs in the pivot sub-list. Given that there can be at most $m$ distinct pivot sub-lists for an unaligned ordering, the alignment function can yield at most $m$ aligned orderings for a given input. We call this set of (at most $m$) possible alignments (for a given unaligned ordering $S_u^l$) the \emph{alignment set} of $S_u^l$.  

Using the concepts defined above, we state the following properties:
\begin{prop}\label{prop1}
$\forall I^l \in \mathcal{I}^l, \forall S^l \in \mathcal{S}^l$, w-score($ I^l$) $\leq$ w-score($S^l$). 
\end{prop}
\begin{prop}\label{prop2}
$\forall S_u^l \in \mathcal{S}_u^l$, the w-score of $S_u^l$ will be no more than the w-score of any list in the alignment set of $S_u^l$. 
\end{prop}
We do not provide technical proofs of these properties but they may be proved by contradiction, by calling on Rules 1 and 2 in the construction of $f'$. Specifically, if the first property is false, it can be shown that Rule 1 was incorrectly applied, while if the second property is false, Rule 2 was incorrectly applied. 

Intuitively, the first property holds because the score assigned to two records in the same internal set is `too large' for Rule 2 to compensate for it\footnote{Since, as we stated earlier, Rule 2 assigns only a specific entry in the f-matrix to an eligible pair, whereas Rule 1 assigns the sum of all entries in the f-matrix to an eligible pair. Recall that f was non-negative}. It is always the case, in an interleaved list, that at least two records sharing the same set ID will not share a window at all. The second property holds for a similar reason (but with relation to Rules 2 and 3), when comparing a stacked unaligned ordering to its aligned version.

We prove the correctness of the oracle for \emph{True} outputs by first observing that every stacked ordering (whether aligned or unaligned) is guaranteed to have w-score at least $T_1$ (Equation \ref{t1}), by an earlier part of the proof. 

If an ordering with w-score at least $k'$ is interleaved, then every stacked ordering also has score at least $k'$ (Property 1). Consider a specific stacked aligned ordering that is the concatenation of sub-lists $S_1^l, \ldots , S_m^l$ and with records in each sub-list $S^l_i$ ordered as $<s^1_i, \ldots ,s^{w-1}_i>$. We can follow the proof showing correctness of the oracle for \emph{False} outputs \emph{backwards} to \emph{derive} Equation \ref{kclaim}, which in turn, implies the existence of a 2-ordering with score at least $k$. 

The same proof can be employed, with only a change in notation, if the ordering is not interleaved but is stacked and aligned. If the ordering is stacked and \emph{unaligned}, we pick an element from the ordering's alignment set (which has a w-score at least as high, by Property 2) and conduct the analysis in the previous sentence. In either case, Equation \ref{kclaim} will be the end result. 

We conclude that the Karp reduction is correct. This proves the theorem that maximum-score w-ordering for any constant $w > 2$ is NP-complete. 

\subsection{Proof of Theorem \ref{approx1ratio}}
By Theorem \ref{claim1}, we know that maximum-score 2-ordering each block individually is necessary and sufficient for achieving max SN, since Algorithm \ref{approx1} assumes a local $f$. Suppose the w-score of the max list (the list on which max SN runs the merge step) is $\Phi^*$. Before performing 2-ordering, we have a list of $u$ blocks,  $<R_{y_1},\ldots,R_{y_u}>$, where each block is a set of records. Let the maximum score of a block $y_i$ be $\Phi^*[y_i], \forall i \in \{i,\ldots,u\}$. By Theorem \ref{claim1} and because $f$ is local, we have:
\begin{equation}\label{lemma}
\Phi^*=\sum_{i=1}^{u}\Phi^*[y_i]
\end{equation}
If we can prove an approximation ratio $\rho$ for each $\Phi^*[y_i]$, where $\rho$ is the approximation ratio of MAX TOUR-TSP, then Equation \ref{lemma} will prove the theorem. 

We drop the second subscript and consider an arbitrary block $R_y$ with $|R_y|=m$ records. In line 8 of Algorithm \ref{approx1}, the auxiliary subroutine \emph{RecordsToGraph} converts $R_y$ into a weighted, complete, undirected graph (with an extra \emph{dummy} vertex).  Consider again Figure \ref{convert}. Since the problem is to find a tour, we can use the dummy vertex, without loss of generality, as the first and last vertex in the \emph{returned} tour (by MAX TOUR-TSP) $<dummy,e_1,v_1,\ldots,v_{m},e_{m+1},dummy>$. We adopt the usual notation that record $r_i$ was bijectively mapped to vertex $v_i$, with $i$ ranging from 1 to $m$. 

The main observation is that, regardless of whether the tour is optimal or not, \emph{any} edge following or preceding the \emph{dummy} vertex will always be 0 by construction, and will not contribute anything to the total tour weight. The difference (between the optimal and sub-optimal tour weights) can only be due to the \emph{path} $<v_1,\ldots,v_{m}>$. Let the same path in the optimal tour be denoted as $path^*$. If the approximation ratio of the TSP algorithm is $\rho$, this implies that:
\begin{equation}
\frac{\sum_{e \in Edges(path)} W(e)}{\sum_{e \in Edges(path^*)} W(e)} \geq \rho
\end{equation} 
The \emph{TourToList} subroutine in Algorithm \ref{approx1} converts this inner path into the ordered list of records, as illustrated in Figure \ref{convert}(b). Thus, the numerator in the summation above is exactly $\Phi[y]$ and the denominator is $\Phi^*[y]$. This proves the theorem.

\subsection{Construction showing score-decline of greedy adjacent swapping}
\begin{figure}[t]
\includegraphics[width=5cm, height=5cm]{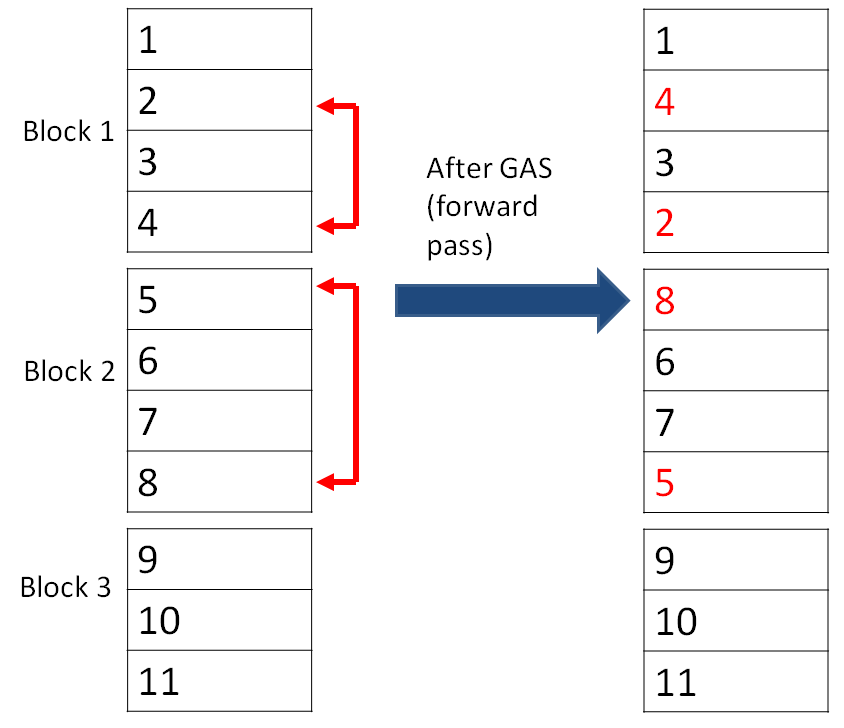}
\caption{A three-block construction showing that greedy adjacent swapping (GAS) can potentially lead to decline in score. $w=2$ is assumed, and it can be verified (using Figure \ref{gasmatrix}) that each block on the left hand side is maximum-score 2-ordered}
\label{decline}
\end{figure}
\begin{figure}[t]
\includegraphics[width=8cm, height=3cm]{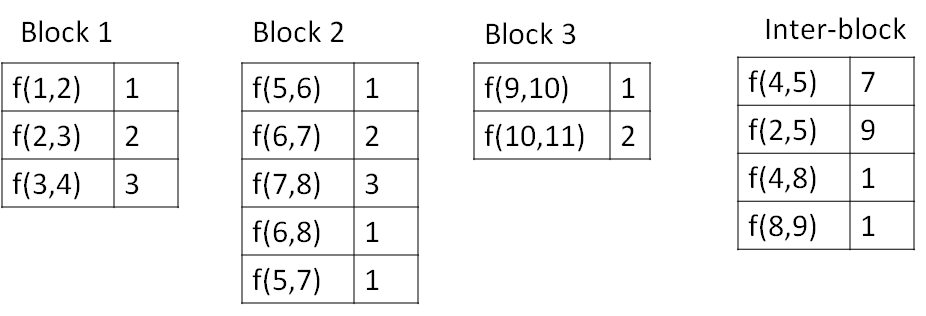}
\caption{The score matrix used in the construction in Figure \ref{decline}. Any scores not in the matrix evaluate to 0}
\label{gasmatrix}
\end{figure}

We show a construction that proves that greedy adjacent swapping (GAS), described in Section \ref{app3}, does not always improve scores. The construction is shown in Figure \ref{decline}. Although we only show the construction for forward pass, a symmetric case can be constructed for the backward pass.

Assuming $w=2$, the list on which the GAS procedure operates has each block maximum-score 2-ordered\footnote{As an additional detail, list \emph{polarities} are also correct, which is required by Algorithm \ref{approx3}. The correctness is evident since $f(4,5)>f(4,8)$ and $f(8,9)>f(8,11)$}, which may be easily verified in Figure \ref{decline} by noting the score matrices for each block in Figure \ref{gasmatrix}. The correctness of the GAS procedure itself can also be verified by using the given scores (in Figure \ref{gasmatrix}) between records belonging to different blocks (\emph{inter-block} scores in the figure). If we now perform merge and compute the candidate set score for the two lists in Figure \ref{decline} (before and after the GAS procedure was run), $\Gamma^{orig}$ is found to have score $6+7+6+1+3=23$, where we break up the score by each block's score (6, 6 and 3) and the scores contributed by records straddling adjacent blocks (7 and 1). Similarly, the candidate set score after GAS, $\Gamma^{GAS}$ is computed as $5+0+1+4+3=13$. The score has declined by a considerable margin.

To understand why such a decline occurred, consider the swap that took place in the second block. If after the first swap (of records 2 and 4 in Block 1), we would have computed the candidate set score, it would have increased, since the increase in score would have been $f(2,5)=9$ according to Figure \ref{gasmatrix} and the decline in the local 2-ordering score of Block 1 would only be $6-5=1$. The global decline occurs because the first and last records in Block 2 end up getting swapped in the next step, which is a valid step according to GAS semantics. Because the procedure is \emph{greedy}, it only takes into account the impact that the swap has on the local score of the \emph{current} block. In other words, the ramifications on \emph{previous} blocks are ignored. 

Given that such declines can occur, we are therefore justified in comparing all three lists in line 7 of Algorithm \ref{approx3}.  
\end{document}